\documentclass[10pt]{article}
\usepackage{logic9}
\usepackage[utf8x]{inputenc}
\usepackage{tikz}
\usetikzlibrary{automata,fit}
\usetikzlibrary{positioning}

\usepackage{graphicx}
\usepackage{listings}
\usepackage{multicol, multirow}
\usepackage{url}

\newcommand{\xra}{\xrightarrow}
\newcommand{\Closure}{\operatorname{\mathit{Closure}}}
\newcommand{\Approx}{\operatorname{\mathit{Approx}}}
\newcommand{\tuple}[1]{\ensuremath{(#1)}}
\newcommand{\Rpos}[0]{\ensuremath{\mathbb{R}_{\geq 0}}}

\newcommand{\CC}[0]{\ensuremath{\Phi}} 
\newcommand{\val}[0]{\ensuremath{\nu}} 
\newcommand{\vali}[0]{\ensuremath{\mathbf{0}}} 
\newcommand{\intpart}[1]{\ensuremath{\lfloor #1\rfloor}}
\newcommand{\fracpart}[1]{\ensuremath{\{ #1\}}}

\newcommand{\reset}[1]{\ensuremath{[#1]}}

\newcommand{\regequiv}[0]{\ensuremath{\sim}}

\newcommand{\clequiv}[2]{\ensuremath{[#1]_{#2}}}

\newcommand{\simu}{simulation } 
\newcommand{\sg}{SG}

\newcommand{\GZG}{GZG^a}
\newcommand{\ZG}{ZG^a}

\newcommand{\qed}[0]{}

\lstdefinelanguage{algo}{%
  morekeywords={procedure,call,push,for,all,and,or,if,then,else,repeat,pop,until,while,do,report}
}

\title{Efficient Emptiness Check for Timed Büchi Automata\\
  (Extended version)}

\author{F. Herbreteau, B Srivathsan and I. Walukiewicz\\
  {\small Univ. Bordeaux, CNRS, LaBRI, UMR 5800, F-33400 Talence,
    France} }

\begin{document}

\maketitle

\begin{abstract}
  The Büchi non-emptiness problem for timed automata refers to deciding
  if a given automaton has an infinite non-Zeno run satisfying the
  Büchi accepting condition. The standard solution to this problem
  involves adding an auxiliary clock to take care of the
  non-Zenoness. In this paper, it is shown that this simple
  transformation may sometimes result in an exponential blowup. A
  construction avoiding this blowup is proposed. It is also shown that
  in many cases, non-Zenoness can be ascertained without extra
  construction. An on-the-fly algorithm for the non-emptiness problem,
  using non-Zenoness construction only when required, is
  proposed. Experiments carried out with a prototype implementation of
  the algorithm are reported.
\end{abstract}

\section{Introduction}

Timed automata~\cite{AD:TCS:1994} are widely used to model real-time
systems. They are obtained from finite automata by adding clocks that
can be reset and whose values can be compared with constants. The
crucial property of timed automata is that their emptiness is
decidable. This model has been implemented in verification tools like
Uppaal~\cite{BDLHPYH:QEST:2006} or Kronos~\cite{BDMOTY:CAV:1998}, and
used in industrial case
studies~\cite{HSLL:RTSS:1997,BBP:IJPR:2004,JRLD:FORMATS:2007}.

While most tools concentrate on the reachability problem, questions
concerning infinite executions of timed automata are also of
interest. In the case of infinite executions one has to eliminate the
so-called Zeno runs. These are executions that contain infinitely many
steps taken in a finite time interval. For obvious reasons such
executions are considered unrealistic. One way to treat Zeno runs
would be to say that a timed automaton admitting such a run is faulty
and should be disregarded. This gives rise to the problem of detecting
the existence of Zeno runs in an
automaton~\cite{Bowman:FAC:2006,GB:FORMATS:2007,Herbreteau:CONCUR:2011}. The
other approach to handling Zeno behaviours, that we adopt here, is to
say that due to imprecisions introduced by the modeling process one
may need to work with automata having Zeno runs. This leads to the
problem of this paper: given a timed automaton decide if it has a
non-Zeno run passing through accepting states infinitely often. We
call this the \emph{Büchi non-emptiness} problem.

This basic problem \cite{AD:TCS:1994} has been studied
already in the paper introducing timed automata. It has been shown
that using so-called region abstraction the problem can be reduced to
the problem of finding a path in a finite region graph satisfying some
particular conditions. The main difference between the cases of finite
and infinite executions is that in the latter one needs to decide if
the path that has been found corresponds to a non-Zeno run of the
automaton.

Subsequent research has shown that the region abstraction is very
inefficient for reachability problems. Another method using zones
instead of regions has been proposed \cite{D:AVMFSS:1990}. It is used
at present in all timed-verification tools. While simple at the first
sight, the zone abstraction was delicate to get
right~\cite{Bouyer:STACS:2003}. This is mainly because the basic
properties of regions do not transfer to zones. The zone abstraction
also works for infinite executions, but unlike for regions, it is
impossible to decide if a path in a zone graph corresponds to a
non-Zeno run of the automaton.

There exists a simple solution to the problem of Zeno runs that
amounts to transforming automata in such way that every run passing
through an accepting state infinitely often is non-Zeno. An automaton
with such a property is called \emph{strongly non-Zeno}. The
transformation is easy to describe and requires the addition of one
new clock. This paper is motivated by our experiments with an
implementation of this construction. We have observed that this
apparently simple transformation can give a big overhead in the size
of a zone graph.

In this paper we closely examine the transformation to strongly
non-Zeno automata~\cite{TYB:FMSD:2005}, and show that it can inflict a
blowup of the zone graph; and this blowup could even be exponential in
the number of clocks. To substantiate, we exhibit an example of an
automaton having a zone graph of polynomial size, whose transformed
version has a zone graph of exponential size. We propose another
solution to avoid this phenomenon. Instead of modifying the automaton,
we modify the zone graph. We show that this modification allows us to
detect if a path in the zone graph can be instantiated to a non-Zeno
run. Moreover the size of the modified graph is $|ZG(\Aa)|\cdot\Oo(|X|)$,
where $|ZG(\Aa)|$ is the size of the zone graph and $|X|$ is the
number of clocks.

In the second part of the paper we propose an on-the-fly algorithm for
testing the existence of accepting non-Zeno runs in timed B\"{u}chi
automata.  The problem we face highly resembles the emptiness testing
of finite automata with generalized Büchi conditions. Since the most
efficient solutions for the latter problem are based on Tarjan's
algorithm to detect strongly-connected-components (SCCs)
\cite{SE:TACAS:2005,GS:MEMICS:2009}, we take the same route here. We
additionally observe that Büchi emptiness can sometimes be decided
directly from the zone graph. This permits to restrict the use of the
modified zone graph construction only to certain parts of the zone
graph. In cases when no clock comparisons of the form $x=0$ are
reachable from the initial state of the automaton, the algorithm runs
in time $\Oo(|ZG(\Aa)|\cdot|X|)$. Further, the optimized algorithm runs in
time $\Oo(|ZG(\Aa)|)$ when no reachable SCC contains a blocking clock:
that is, a clock that is bounded (e.g. $x\leq 1$) but never reset in
the SCC. We also give additional optimizations that prove to be
powerful in practice. We include experiments conducted on examples in
the literature.

\subsection{Related work}
The zone approach has been introduced in the Petri net
context~\cite{BM:IFIP:1983}, and then adapted to the framework of
timed automata~\cite{D:AVMFSS:1990}. The advantage of zones over
regions is that they do not require to consider every possible unit
time interval separately. The delicate point about zones was to find a
right approximation operator. Usual approximation operators are sound
and complete: each path in the zone graph can be instantiated as a run
in the automaton and vice-versa. While this is enough for correctness
of the reachability algorithm, it does not allow however to determine
if a path can be instantiated to a non-Zeno run. The solution
involving adding one clock has been discussed
in~\cite{T:ARTS:1999,TYB:FMSD:2005,AM:SFM-RT:2004}. Recently,
Tripakis~\cite{T:TOCL:2009} has shown a way to extract an accepting
run from a zone graph of the automaton. Combined with the construction
of adding one clock this gives a solution to the B\"{u}chi emptiness
problem. Since, as we show here, adding one clock may be costly, this
solution is costly too. A different approach has been considered
in~\cite{Bowman:FAC:2006,GB:FORMATS:2007} where some sufficient
conditions are proposed for a timed automaton to be free from Zeno
runs. Notice that for obvious complexity reasons, any such condition
must be either not complete, or of the same algorithmic complexity as
the emptiness test itself.

\subsection{Organization of the paper}
In the next section we formalize our problem, and discuss region and
zone abstractions. As an intermediate step we give a short proof of
the above mentioned result
from~\cite{T:TOCL:2009}. Section~\ref{sec:nonzeno} explains the
problems with the transformation to strongly non-Zeno automata, and
describes our alternative method. The following section is devoted to
a description of the algorithm. We conclude with the results of the
experiments performed.



\section{The Emptiness Problem for Timed Büchi Automata}
\label{sec:empty_tba}

\subsection{Timed Büchi Automata}
\label{sec:empty_tba:tba}

Let $X$ be a set of clocks, i.e., variables that range over $\Rpos$,
the set of non-negative real numbers. \emph{Clock constraints} are
conjunctions of comparisons of variables with integer constants: $x\#
c$ where $x\in X$ is a clock, $c\in\Nat$ and
$\#\in\{<,\leq,=,\geq,>\}$. For instance $(x\leq 3\land y>0)$ is a
clock constraint. Let $\CC(X)$ denote the set of clock constraints
over clock variables $X$.

A \emph{clock valuation} over $X$ is a function
$\val\,:\,X\rightarrow\Rpos$. We denote $\Rpos^X$ for the set of clock
valuations over $X$, and $\vali\,:\,X\rightarrow\{0\}$ for the valuation
that associates $0$ to every clock in $X$. We write $\val\models \phi$
when $\val$ satisfies $\phi$, i.e.  when every constraint in $\phi$
holds after replacing every $x$ by $\val(x)$.

For a valuation $\val$ and  $\delta\in\Rpos$, let $(\val+\delta)$
be the valuation such that $(\val+\delta)(x)=\val(x)+\delta$ for all
$x\in X$. For a set $R\subseteq X$, let 
$\reset{R}\val$ be the valuation such that $(\reset{R}\val)(x)=0$
if $x\in R$ and $(\reset{R}\val)(x)=\val(x)$ otherwise.

A \emph{Timed Büchi Automaton (TBA)} is a tuple
$\mathcal{A}=\tuple{Q,q_0,X,T,\Acc}$ where $Q$ is a finite set of
states, $q_0\in Q$ is the initial state, $X$ is a finite set of
clocks, $\Acc\subseteq Q$ is a set of accepting states, and
$T\,\subseteq\, Q\times\CC(X)\times 2^X \times Q$ is a finite set of
transitions $\tuple{q,g,R,q'}$ where $g$ is a \emph{guard}, and $R$ is
a \emph{reset} of the transition.

A \emph{configuration} of $\Aa$ is a pair $\tuple{q,\val}\in
Q\times\Rpos^X$; with $\tuple{q_0,\vali}$ being the \emph{initial
  configuration}. A \emph{transition} $\tuple{q,\val}\xra{\d,t}
\tuple{q',\val'}$ for $t=(q,g,R,q')\in T$ and $\d\in\Rpos$ is defined
when $\val+\d\sat g$ and $\val'=\reset{R}(\val+\d)$.

A \emph{run} of $\mathcal{A}$ is an infinite sequence of
configurations connected by transitions, starting from the initial
state $q_0$ and the initial valuation $\val_0=\vali$:
\begin{equation*}
  \tuple{q_0,\val_0}\xrightarrow{\delta_0,t_0}
  \tuple{q_1,\val_1}\xrightarrow{\delta_1,t_1}
  \cdots
\end{equation*}
A run $\sigma$ \emph{satisfies the Büchi condition} if it visits
\emph{accepting configurations} infinitely often, that is
configurations with a state from $\Acc$.  The \emph{duration} of the
run is the accumulated delay: $\sum_{i\geq 0} \delta_i$.  An infinite run
$\sigma$ is \emph{Zeno} if it has a finite duration.

\begin{definition}
  \label{defn:language_tba}
  The \emph{Büchi non-emptiness problem} is to decide if $\Aa$ has a
  non-Zeno run satisfying the Büchi condition.
\end{definition}

The Büchi non-emptiness problem is known to be
\PSPACE-complete~\cite{AD:TCS:1994}.

The class of TBA we consider is usually known as diagonal-free TBA
since clock comparisons like $x-y\leq 1$ are disallowed.  Since we are
interested in the Büchi non-emptiness problem, we can consider
automata without an input alphabet and without invariants since they
can be simulated by guards.

\subsection{Regions and region graphs}
\label{sec:empty_tba:regions}

A simple decision procedure for the Büchi non-emptiness problem builds
from $\Aa$ a graph called the \emph{region graph} and tests if there is a
path in this graph satisfying certain conditions. We will define two
types of regions.

Fix a constant $M$ and a finite set of clocks $X$. Two valuations
$\nu,\nu'\in \Rpos^X$ are \emph{region equivalent} w.r.t. $M$, denoted
$\val\regequiv_{M} \val'$ iff for every $x,y\in X$:
\begin{enumerate}
\item $\nu(x)>M$ iff $\nu'(x)>M$;
\item if $\nu(x)\leq M$, then $\intpart{\val(x)}=\intpart{\val'(x)}$;
\item if $\nu(x)\leq M$, then $\fracpart{\val(x)}=0$ iff
  $\fracpart{\val'(x)}=0$;
\item if $\nu(x)\leq M$ and $\nu(y)\leq M$ then
  $\fracpart{\val(x)}\leq \fracpart{\val(y)}$ iff
$\fracpart{\val'(x)}\leq \fracpart{\val'(y)}$. 
\end{enumerate}

The first three conditions ensure that the two valuations satisfy the
same guards as clock constraints are defined with respect to integer
bounds and $M$ is the maximal constant in $\Aa$. The last one enforces
that for every $\d\in\Rpos$ there is $\d'\in\Rpos$, such that
valuations $\val+\d$ and $\val'+\d'$ satisfy the same guards since the
difference of $x$ and $y$ is invariant by time elapse.

We will also define \emph{diagonal region equivalence (d-region
  equivalence} for short) that strengthens the last condition to
\begin{description}
\item [$4^d$.] for every integer $c\in (-M,M)$: $\val(x)-\val(y)\leq
  c$ iff $\val'(x)-\val'(y)\leq c$
\end{description}
This region equivalence is denoted by $\regequiv^d_{M}$. Observe that it
 is finer than $\regequiv_M$.

A \emph{region} is an equivalence class of $\regequiv_{M}$. We write
$\clequiv{\val}{\regequiv_{M}}$ for the region of $\val$, and $\Rr_M$
for the set of all regions with respect to $M$.  Similarly, for
d-region equivalence we write: $\clequiv{\val}{\regequiv_{M}}^d$ and
$\Rr_M^d$.  If $r$ is a region or a d-region then we will write $r\sat
g$ to mean that every valuation in $r$ satisfies the guard
$g$. Observe that all valuations in a region, or a d-region, satisfy the
same guards.

For an automaton $\Aa$, we define its \emph{region graph}, $RG(\Aa)$,
using the $\regequiv_M$ relation, where $M$ is the biggest constant
appearing in the guards of its transitions. Without loss of generality
we assume that $M \ge 0$, in other words there is at least one guard
in $\Aa$. Nodes of $RG(\Aa)$ are of the form $(q,r)$ for $q$ a state
of $\Aa$ and $r\in \Rr_M$ a region. There is a transition
$(q,r)\xra{t}(q',r')$ if there are $\val\in r$, $\d\in\Rpos$ and
$\val'\in r'$ with $(q,\val)\xra{\d,t}(q',\val')$. Observe that a
transition in the region graph is not decorated with a delay.  The
graph $RG^d(\Aa)$ is defined similarly but using the $\regequiv^d_M$
relation.

It will be important to understand the properties of pre- and
post-stability of regions or d-regions~\cite{TYB:FMSD:2005}. We state
them formally. A transition $(q,r)\xra{t}(q',r')$ in a region graph or
a d-region graph is:
\begin{itemize}
\item \emph{Pre-stable} if for every $\val\in r$ there are $\val'\in
  r'$, $\delta \in \Rpos$ s.t. $(q,\val)\xra{\delta,t}(q',\val')$.
\item \emph{Post-stable} if for every  $\val'\in r'$ there are $\val\in
  r$, $\delta \in \Rpos$ s.t. $(q,\val)\xra{\delta,t}(q',\val')$.
\end{itemize}

The following lemma explains our interest in $\regequiv^d_M$
relation. The main fact is that both region graphs are pre-stable and
this allows to decide the existence of a non-Zeno run easily by
Theorem~\ref{thm:empty-region-graph}.

\begin{lemma}[Pre and post-stability~\cite{B:FMSD:2004}]
Transitions in $RG^d(\Aa)$ are pre-stable and post-stable. Transitions
in $RG(\Aa)$ are pre-stable but not necessarily post-stable. 
\end{lemma}

Consider two sequences
\begin{gather}\label{eq:run}
\tuple{q_0,\val_0}\xrightarrow{\delta_0,t_0}
    \tuple{q_1,\val_1}\xrightarrow{\delta_1,t_1}
    \cdots\\
\label{eq:region-path}
    \tuple{q_0,r_0}
    \xrightarrow{t_0}
    \tuple{q_1,r_1}\xrightarrow{t_1}
    \cdots 
\end{gather}
where the first is a run in $\Aa$, and the second is a path in
$RG(\Aa)$ or $RG^d(\Aa)$. We say that the first is an \emph{instantiation}
of the second if $\val_i\in r_i$ for all $i\geq 0$. Equivalently, we
say that the second is an \emph{abstraction} of the first. The
following lemma is a direct consequence of the pre-stability property.
\begin{lemma}
  \label{lem:region-run-abstraction}
  Every path in $RG(\Aa)$ is an abstraction of a run of $\Aa$, and
  conversely, every run of $\Aa$ is an instantiation of a path in
  $RG(\Aa)$. Similarly for $RG^d(\Aa)$.
\end{lemma}

This lemma allows us to relate the existence of an accepting run of
$\Aa$ to the existence of paths with special properties in $RG(\Aa)$
or $RG^d(\Aa)$. We say that a path as in~\eqref{eq:region-path}
\emph{satisfies the Büchi condition} if it has infinitely many
occurrences of states from $\Acc$. The path is called
\emph{progressive}~\cite{AD:TCS:1994,TYB:FMSD:2005} if for every clock
$x\in X$:\label{progressive}
\begin{itemize}
\item either $x$ is almost always above $M$: there is $n$ with
  $r_i\sat x>M$ for all $i>n$;
\item or $x$ is reset infinitely often and strictly positive
  infinitely often: for every $n$ there are $i,j>n$ such that
  $r_i\sat(x=0)$ and $r_j\sat(x>0)$.
\end{itemize}

\begin{theorem}[\cite{AD:TCS:1994}]
  \label{thm:empty-region-graph}
  A TBA $\mathcal{A}$ has a non-Zeno run satisfying the Büchi
  conditions iff $RG(\Aa)$ has a progressive path satisfying the Büchi
  condition. Similarly for $RG^d(\Aa)$.
\end{theorem}

The progress criterion above can be encoded adding an extra Büchi
accepting condition~\cite{AD:TCS:1994,TYB:FMSD:2005}. While
theorem~\ref{thm:empty-region-graph} gives an algorithm for solving
our problem, it turns out that this method is very impractical. The
number of regions for clocks $X$ and constant $M$ turns out to be
$\mathcal{O}(|X|!.2^{|X|}M^{|X|})$~\cite{AD:TCS:1994} and constructing
all of them, or even searching through them on-the-fly, has proved to
be very costly.

\subsection{Zones and zone graphs}
\label{sec:zones}
Timed verification tools use zones instead of regions.  A zone is a
set of valuations defined by a conjunction of two kinds of
constraints: comparison of the difference between two clocks with a
constant like $x-y\# c$, or comparison of the value of a single clock
with a constant like $x\# c$ for $x\in X$, $c\in\Nat$ and
$\#\in\{<,\leq,=,\geq,>\}$. For example $(x-y\geq 1)\land(y<2)$ is a
zone. While at first sight it may seem that there are more zones than
regions, this is not the case if we count only those that are
reachable from the initial valuation.

Since zones are sets of valuations defined by constraints, one can
define transitions directly on zones. For a transition $t$ in $\Aa$
and a zone $Z$, we have $(q,Z)\xra{t} (q',Z')$ if $Z'$ is the set of
valuations $\val'$ such that there exists $\val\in Z$ and $\d\in\Rpos$
and $(q,\val)\xra{\d,t}(q',\val')$. It is well-known that $Z'$ is a
zone. Moreover zones can be represented using Difference Bound
Matrices (DBMs), and transitions can be computed efficiently on
DBMs~\cite{D:AVMFSS:1990}. The problem is that the number of reachable
zones is not guaranteed to be finite~\cite{DT:TACAS:1998}.

In order to ensure that the number of reachable zones is finite, one
introduces abstraction operators. We mention the three most common
ones in the literature.  They refer to region graphs, $RG(\Aa)$ or
$RG^d(\Aa)$, and use the constant $M$ that is the maximal constant
appearing in the guards of $\Aa$.
\begin{itemize}
\item $\Closure_M(Z)$: the smallest union of regions containing $Z$;
\item $\Closure^d_M(Z)$: similarly but for d-regions;
\item $\Approx_M(Z)$: the smallest union of d-regions that is convex
  and that contains $Z$.
\end{itemize}

The following lemma establishes the links between the three
abstraction operators, and is very useful to transpose reachability
results from one abstraction to the other.
\begin{lemma}[\cite{B:FMSD:2004}]\label{lemma:abstraction-inclusion}
  For every zone $Z$: $Z\incl \Closure^d_M(Z)\incl
  \Approx_M(Z)\incl\Closure_M(Z)$.
\end{lemma}

Similar to region graphs, we define simulation graphs where after
every transition a specific approximation operation is used. So we
have three graphs corresponding to the three approximation operations
above. Notice that $\Closure_M(Z)$ and $\Closure^d_M(Z)$ may not be
convex, hence they may not be zones~\cite{B:FMSD:2004}.

Take an automaton $\Aa$ and let $M$ be the biggest constant that
appears in the guards of its transitions. The simulation graph
$\sg(\Aa)$ has nodes of the form $(q,S)$ where $q$ is a state of $\Aa$
and $S$ is a set of valuations. The initial node is
$(q_0,\{\vali\})$. There is a transition $(q,S)\xra{t}
(q',\Closure_M(S'))$ in $\sg(\Aa)$ if $S'$ is the set of valuations
$\val'$ such that $(q,\val)\xra{\d,t}(q',\val')$ for some $\val\in S$
and $\d\in\Rpos$. Similarly, we define $\sg^d(\Aa)$ and $\sg^a(\Aa)$
by replacing $\Closure_M$ with $\Closure_M^d$ and $\Approx_M$
respectively. Observe that for every node $(q,S)$ that is reachable in
one of the three graphs above, $S$ is a union of regions or
d-regions. The notions of an abstraction of a run of $\Aa$, and an
instantiation of a path in the simulation graph, are defined in the same
way as that of region graphs.

Tools like Kronos or Uppaal use the $\Approx_M$ abstraction. The two other
abstractions are less interesting for implementations since the result
may not be convex. Nevertheless, they are useful in proofs. The
following Lemma (cf.~\cite{DT:TACAS:1998}) says that transitions in
$\sg(\Aa)$ are post-stable with respect to regions.
\begin{lemma}\label{ZG-region-post-stable}
  Let $(q,S)\xra{t}(q',S')$ be a transition in $\sg(\Aa)$ such that
  both $S$ and $S'$ are unions of regions. For every region $r'\incl
  S'$, there is a region $r\incl S$ such that $(q,r)\xra{t}(q',r')$ is
  a transition in $RG(\Aa)$.
\end{lemma}

\begin{proof}
  Take a transition $(q,S)\xra{t}(q',S')$ and let us examine what it
  means. By definition, $S'=\Closure_M(S'')$ where $S''$ is the set of
  valuations $\val''$ that satisfy $(q,\val)\xra{\d,t}(q',\val'')$ for
  some $\val\in S$ and $\d\in\Rpos$. Consider $r'\subseteq S'$; the
  intersection $r'\cap S''$ is not empty. Take $\val'\in r'\cap S''$,
  and let $\val\in S$ be a valuation such that
  $(q,\val)\xra{\d,t}(q',\val')$ for some $\d\in\Rpos$. Let $r$ be the
  region of $\val$. We have that $r\cap S$ is not empty, hence
  $r\subseteq S$ as $S$ is a union of regions. By definition
  $(q,r)\xra{t}(q',r')$ is a transition in $RG(\Aa)$.
  \qed
\end{proof}

We get a correspondence between paths in \simu graphs and
runs of $\Aa$.

\begin{theorem}[\cite{T:TOCL:2009}]
  \label{thm:concretization_tri09}
  Every path in $\sg(\mathcal{A})$ is an abstraction of a run of $\Aa$,
  and conversely, every run of $\Aa$ is an instantiation of a path in
  $\sg(\mathcal{A})$. Similarly for $\sg^d$ and $\sg^a$.
\end{theorem}

\begin{proof}
  We first show that a path in $\sg(\Aa)$ is an abstraction of a run of
  $\Aa$.  Take a path $(q_0,S_0)\xra{t_0}(q_1,S_1)\xra{t_1}\dots$ in
  $\sg(\Aa)$. Construct a DAG with nodes $(i,q_i,r_i)$ such that $r_i$
  is a region in $S_i$. We put an edge from $(i,q_i,r_i)$ to
  $(i+1,q_{i+1},r_{i+1})$ if $(q_i,r_i)\xra{t_i}(q_{i+1},r_{i+1})$. By
  Lemma~\ref{ZG-region-post-stable}, every node in this DAG has at
  least one predecessor, and the branching of every node is bounded by
  the number of regions. Hence, this DAG has an infinite path that is
  a path in $RG(\Aa)$. By Lemma~\ref{lem:region-run-abstraction} this
  path can be instantiated to a run of $\Aa$.
  
  To conclude the proof one can easily verify that a run of $\Aa$ can
  be abstracted to a path in $\sg^d(\Aa)$. Then using
  Lemma~\ref{lemma:abstraction-inclusion} this path can be converted
  to a path in $\sg^a(\Aa)$, and later to one in $\sg(\Aa)$.  \qed
\end{proof}

\begin{figure}
  \begin{tikzpicture}[line width=1pt]
    \begin{scope}
      \node (z0_0) at (0,0) {\footnotesize $(0,0=x=y)$};
      \node (z1_0) at (2.5,0) {\footnotesize $(1,0=x=y)$};
      \node (z1_0') at (2.5,1) {\footnotesize $(1,0=x<y)$};
      \node (z0_1) at (5,0) {\footnotesize $(0,0=x=y)$};
      \node (z2_0) at (7.5,0) {\footnotesize $(2,0=y=x)$};
      \node (z2_0') at (7.5,1) {\footnotesize $(2,0=y<x)$};
      \node (z0_0') at (10,0) {\footnotesize $(0,0=x=y)\ \cdots$};
    \end{scope}
    \begin{scope}[->]
      \draw (z0_0) -- (z1_0);
      \draw (z0_0) -- (z1_0');
      \draw (z1_0) -- (z0_1);
      \draw (z0_1) -- (z2_0);
      \draw (z0_1) -- (z2_0');
      \draw (z2_0) -- (z0_0');
    \end{scope}
  \end{tikzpicture}
  \caption{A part of the region graph for the automaton $\Aa_2$ in
    Figure~\ref{fig:gzg-eg}.}
  \label{fig:part-RG-A2}
\end{figure}

\begin{figure}
  \begin{tikzpicture}[line width=1pt]
    \begin{scope}
      \node (z0_0) at (0,0) {\footnotesize $(0,0=x=y)$};
      \node (z1_0) at (2.5,0) {\footnotesize $(1,0=x\leq y)$};
      \node (z0_1) at (5,0) {\footnotesize $(0,0=x=y)$};
      \node (z2_0) at (7.5,0) {\footnotesize $(2,0=y\leq x)$};
      \node (z0_0') at (10,0) {\footnotesize $(0,0=x=y)\ \cdots$};
    \end{scope}
    \begin{scope}[->]
      \draw (z0_0) -- (z1_0);
      \draw (z1_0) -- (z0_1);
      \draw (z0_1) -- (z2_0);
      \draw (z2_0) -- (z0_0');
    \end{scope}
  \end{tikzpicture}
  \caption{A part of the symbolic graph for the automaton $\Aa_2$ in
    Figure~\ref{fig:gzg-eg}.}
  \label{fig:part-SG-A2}
\end{figure}

Observe that Theorem~\ref{thm:concretization_tri09} does not guarantee
that a path we find in a simulation graph has an instantiation that is
non-Zeno. This cannot be decided from $\sg(\Aa)$ by using the progress
criterion defined in page~\pageref{progressive} as we show
now. Consider for instance the automaton $\Aa_2$ in
Figure~\ref{fig:gzg-eg} which has only Zeno runs as both $x$ and $y$
must remain equal to $0$ on every run. Figure~\ref{fig:part-RG-A2}
shows a part of $RG(\Aa_2)$. The infinite path starting from node
$(0,0=x=y)$ is not progressive as none of the clocks can have a
positive value. Moreover, it can be seen that every node where a clock
has a positive value is a deadlock node. Figure~\ref{fig:part-SG-A2}
depicts the corresponding part of $\sg(\Aa_2)$. This path satisfies
the progress criterion as both $x$ and $y$ are reset and may have
positive values infinitely often, despite all its instantiations being
Zeno. The progress criterion fails due to the loss of pre-stability in
$SG(\Aa_2)$: none of the valuations with either $x>0$ or $y>0$ have a
successor. In Section~\ref{sec:nonzeno}, we show how to avoid this
problem.

In the subsequent sections, we are interested only in the simulation
graph $SG^a(\Aa)$. Observe that the symbolic zone obtained by the
approximation of a zone using $\Approx_M$ is in fact a zone.  Hence,
we prefer to call it a zone graph and denote it as $ZG^a(\Aa)$. Every
node of $ZG^a(\Aa)$ is of the form $(q,Z)$ where $Z$ is a zone.



\section{Finding non-Zeno paths}
\label{sec:nonzeno}

As we have remarked above, in order to use
Theorem~\ref{thm:concretization_tri09} we need to be sure that a path
we get can be instantiated to a non-Zeno run. We discuss the solutions
proposed in the literature, and then offer a better one. Thanks to
pre-stability of the region graph, the progress criterion on regions
has been defined in~\cite{AD:TCS:1994} for selecting runs from
$RG(\Aa)$ that have a non-Zeno instantiation (see
Section~\ref{progressive}). Notice that the semantics of TBA
in~\cite{AD:TCS:1994} constrains all delays $\delta_i$ to be strictly
positive, but the progress criterion can be extended to the stronger
semantics that is used nowadays (see~\cite{TYB:FMSD:2005} for
instance). However, since zone graphs are not pre-stable, this method
cannot be directly extended to zone graphs.

\subsection{Adding one clock}
\label{sec:non_zeno:adding_one_clock}

A common solution to deal with Zeno runs is to transform an automaton
into a \emph{strongly non-Zeno automaton}, i.e. such that all runs
satisfying the Büchi condition are guaranteed to be non-Zeno. We
present this solution here and discuss why, although simple, it may
add an exponential factor in the decision procedure.

\begin{figure}
  \begin{tikzpicture}[shorten >=1pt,node distance=2cm,on grid,auto]
    \node[state] (b) {$b^k$};
    \node[state,accepting] (a) [above=of b] {$a^k$};
    \node[node distance=1.5cm] (predots) [left=of b] {$\dots$};
    \node[node distance=1.5cm] (postdots) [right=of b] {$\dots$};
    \draw[->] (b) edge [bend left] node {\footnotesize $y\leq d$} (a);
    \draw[->] (a) edge [bend left] node {\footnotesize
      $\{x_1,\dots,x_{k-1}\}$} (b);
    \draw[->] (predots) edge (b);
    \draw[->] (b) edge (postdots);
    \node[node distance=1.2cm] (figname) [below=of b]
    {$\mathbf{\Vv_k}$};
  \end{tikzpicture}
  \hfill
  \begin{tikzpicture}[shorten >=1pt,node distance=3cm,on grid]
    \node[state] (b) {$b^k$};
    \node[state,accepting] (a') [above left=of b] {${a^k_1}$};
    \node[state] (a) [above right=of b] {$a^k_2$};
    \node[node distance=1.5cm] (predots) [left=of b] {$\dots$};
    \node[node distance=1.5cm] (postdots) [right=of b] {$\dots$};
    \draw[->] (b) edge node [left,xshift=-2mm] {\footnotesize
      $\begin{array}{l}y\leq d\ \wedge\\ z\geq 1\end{array}$} node
    [left,xshift=-1mm,near start] {\footnotesize $\{z\}$} (a');
    \draw[->] (b) edge [bend right] node [right,xshift=2mm]
    {\footnotesize $\begin{array}{l}y\leq d\end{array}$}
    (a);
    \draw[->] (a') edge (a);
    \draw[->] (a) edge node [near start,left,xshift=-1mm]
    {\footnotesize $\{x_1,\dots,x_{k-1}\}$} (b);
    \draw[->] (predots) edge (b);
    \draw[->] (b) edge (postdots);
    \node[node distance=1.2cm] (figname) [below=of b]
    {$\mathbf{\Ww_k}=SNZ(\Vv_k)$};
  \end{tikzpicture}
  \caption{The gadgets $\Vv_k$ (left) and $\Ww_k=SNZ(\Vv_k)$ (right).}
  \label{fig:gadget-Vk-Wk}
\end{figure}

The main idea behind the transformation of $\Aa$ into a strongly
non-Zeno automaton $SNZ(\Aa)$ is to ensure that on every accepting
run, time elapses for $1$ time unit infinitely often. Hence, it is
sufficient to check for the existence of an accepting run as it is
non-Zeno for granted. Consider the automaton $\Vv_k$ and its
transformation into $\Ww_k=SNZ(\Vv_k)$ in
Figure~\ref{fig:gadget-Vk-Wk}. The transformation adds one clock $z$
and duplicates accepting states (e.g. $a^k$ in $\Vv_k$). One copy is
no longer accepting whereas the other is accepting, but it can be
reached only when $z\ge 1$ (these are respectively $a_2^k$ and $a_1^k$
in $\Ww_k$). Moreover, when an accepting state is reached $z$ is reset
to $0$. As a result, every accepting run in $\Vv_k$ has a
corresponding run in $\Ww_k$ where every occurrence of $a^k$ is
replaced by an occurrence of either $a_1^k$ or $a_2^k$. Since two
occurrences of the accepting state $a_1^k$ have to be separated by at
least one time unit, an accepting run in $\Ww_k$ is necessarily
non-Zeno.

A slightly different construction is mentioned
in~\cite{AM:SFM-RT:2004}. Of course one can also have other
modifications, and it is impossible to treat all the imaginable
constructions at once. Our objective here is to show that the
constructions proposed in the literature produce a phenomenon causing
proliferation of zones that can sometimes be exponential in the number
of clocks.  The discussion below will focus on the construction
from~\cite{TYB:FMSD:2005}, but the one from~\cite{AM:SFM-RT:2004}
suffers from the same problem.

\begin{figure}
  \footnotesize \centering
  \begin{tikzpicture}[->]
    \tikzstyle{every node}=[rounded corners=2mm,rectangle]
    \node[draw] (b1) at (0,2) 
    {$\begin{array}{c}b^2\\ 0\leq y\leq x_1\leq x_2\\
        \ \end{array}$}; 
    \node[draw] (a1) at (5,2)
    {$\begin{array}{c}a^2\\ 0\leq y\leq x_1\leq x_2\\ \wedge\ y\leq
        d\end{array}$};
    \node[draw] (b2) at (0,0)
    {$\begin{array}{c}b^2\\ 0= x_1\leq y\leq x_2\\ \ \end{array}$};
    \node[draw] (a2) at (5,0) 
    {$\begin{array}{c}a^2\\ 0\leq x_1\leq
        y\leq x_2\\ \wedge\ y\leq d\end{array}$};
    \node (predots_b1) at (-2.5,2.5) {$\dots$};
    \node (postdots_b1) at (-2.5,1.5) {$\dots$};
    \node (postdots_b2) at (-2.5,0) {$\dots$};
    \draw (b1) -- node[above] {$y\leq d$} (a1);
    \draw (a1) .. controls +(270:1cm) and +(90:1cm) .. 
    node[above] {$\{x_1\}$} (b2);
    \draw (b2) -- node[above] {$y\leq d$} (a2);
    \draw (a2) .. controls +(270:1cm) and +(270:1cm) ..
    node[above] {$\{x_1\}$} (b2);
    \draw (predots_b1) -- (b1);
    \draw (b1) -- (postdots_b1);
    \draw (b2) -- (postdots_b2);
  \end{tikzpicture}
  \caption{Part of $ZG(\Vv_2)$}
  \label{fig:zgv}
\end{figure}

\begin{figure}
  \begin{tikzpicture}[shorten >=1pt,node distance=2cm,on grid,auto]
    \node (predots) {$\dots$};
    \node[state,node distance=1.5cm] (c_0) [right=of predots]
    {$c^k_0$};
    \node[state] (c_1) [right=of c_0] {$c^k_1$};
    \node (middledots) [right=of c_1] {$\dots$};
    \node[state] (c_k) [right=of middledots] {$c^k_k$};
    \node[state] (c_y) [right=of c_k] {$c^k_y$};
    \node[node distance=1.5cm] (postdots) [right=of c_y] {$\dots$};
    \draw[->] (predots) edge (c_0);
    \draw[->] (c_0) edge node {\footnotesize $\{x_k\}$} (c_1);
    \draw[->] (c_1) edge node {\footnotesize $\{x_{k-1}\}$}
    (middledots);
    \draw[->] (middledots) edge node {\footnotesize $\{x_1\}$} (c_k);
    \draw[->] (c_k) edge node {\footnotesize $\{y\}$} (c_y);
    \draw[->] (c_y) edge (postdots);
  \end{tikzpicture}
  \caption{The gadget $\Rr_k$.}
  \label{fig:gadget-Rk}
\end{figure}

The problem comes from the fact that the constraint $z\geq 1$ may be a
source of rapid multiplication of the number of zones in the zone
graph of $SNZ(\Aa)$. Consider $\Vv_k$ and $\Ww_k$ from
Figure~\ref{fig:gadget-Vk-Wk} and let us say that $k=2$. Starting at
the state $b^2$ of $\Vv_2$ in the zone $0\leq y\leq x_1\leq x_2$,
there are two reachable zones with state $b^2$. This is depicted in
Figure~\ref{fig:zgv} where after two traversals of the cycle formed by
$b^2$ and $a^2$, we reach a zone that is invariant for the
cycle. Moreover, from the two zones with state $b^2$ in
Figure~\ref{fig:zgv}, reseting $x_1$ followed by $y$ as $\Rr_1$ (in
Figure~\ref{fig:gadget-Rk}) does, we reach the same zone $0\leq y\leq
x_1\leq x_2$.

In contrast starting in $b^2$ of $\Ww_2=SNZ(\Vv_2)$ from $0\leq y\leq
x_1\leq x_2 \leq z$ gives at least $d$ zones. The part of $ZG(\Ww_2)$
in Figure~\ref{fig:zgw} gives the sequence of transitions in the zone
graph of $\Ww_2$ starting from the zone $(b^2, 0 \le y \le x_1 \le x_2
\le z)$ by successive iterations of the cycle that goes through $b^2$,
$a_1^2$ and $a_2^2$. After a certain point, every traversal induces an
extra distance between the clocks $y$ and $z$. Clearly, there are at
least $d$ zones in this case. Resetting $x_1$ followed by $y$ as
$\Rr_1$ (in Figure~\ref{fig:gadget-Rk}) does still yield $d$ zones

\begin{figure}
  \scriptsize \centering
  \begin{tikzpicture}[->]
    \tikzstyle{every node}=[rounded corners=2mm,rectangle,inner
    sep=0pt]
    \node(predots_b1) at (-2.2,0.35) {$\dots$};
    \node(postdots_b1) at (-2.2,-0.35) {$\dots$};
    \node[draw] (b1) at (0,0)
    {$\begin{array}{c}b^2\\ 0\leq y\leq x_1\leq x_2\leq z\\
        \ \end{array}$};
    \node[draw] (a11) at (4,0)
    {$\begin{array}{c}a_1^2\\ 0=z\leq y\leq x_1\leq x_2\\ \wedge\
        y-z\geq 0\,\wedge\, y\leq d\end{array}$};
    \node[draw] (a21) at (8,0)
    {$\begin{array}{c}a_2^2\\ 0\leq z\leq y\leq x_1\leq x_2\\
        y-z\geq 0 \end{array}$};
    \node (postdots_b2) at (-2.2,-2) {$\dots$};
    \node[draw] (b2) at (0,-2)
    {$\begin{array}{c}b^2\\ 0= x_1\leq z\leq y\leq x_2\\
        \ \end{array}$};
    \node[draw] (a12) at (4,-2)
    {$\begin{array}{c}a_1^2\\ 0=z\leq x_1\leq y\leq x_2\\
        \wedge\,y-z\geq 1\,\wedge\,y\leq d\end{array}$};
    \node[draw] (a22) at (8,-2) {$\begin{array}{c}a_2^2\\ 0\leq z\leq
        x_1\leq y\leq x_2\\ \wedge\ y-z\geq 1\end{array}$};
    \node (postdots_b3) at (-2.2,-4) {$\dots$};
    \node[draw] (b3) at (0,-4)
    {$\begin{array}{c}b^2\\ 0= x_1\leq z\leq y\leq x_2\\ \wedge\
        y-z\geq 1 \end{array}$};
    \node[minimum width=2cm,minimum height=5mm] (a13) at (4,-4)
    {$\dots$};
    \node[minimum width=2cm,minimum height=5mm] (a23) at (8,-4)
    {$\dots$};
    \node (postdots_b4) at (-2.2,-6) {$\dots$};
    \node[draw] (b4) at (0,-6)
    {$\begin{array}{c}b^2\\ 0= x_1\leq z\leq y\leq x_2\\ \wedge\
        y-z\geq d-1\end{array}$};
    \node[draw] (a14) at (4,-6)
    {$\begin{array}{c}a_1^2\\ 0=z\leq x_1\leq y\leq x_2\\
        \wedge\,y-z=d\,\wedge\,y=d\end{array}$};
    \node[draw] (a24) at (8,-6)
    {$\begin{array}{c}a_2^2\\ 0\leq z\leq x_1\leq y\leq x_2\\
        \wedge\ y-z=d\end{array}$};
    \node (postdots_b5) at (-2.2,-8) {$\dots$};
    \node[draw] (b5) at (0,-8)
    {$\begin{array}{c}b^2\\ 0= x_1\leq z\leq y\leq x_2\\ \wedge\
        y-z=d\end{array}$};
    \draw (predots_b1) edge (b1);
    \draw (b1) edge (postdots_b1);
    \draw (b1) edge node [above] {\tiny $\begin{array}{r}y\leq d\\
        \wedge\,z\geq 1\end{array}$} node [below] {\tiny $\{z\}$}
    (a11);
    \draw (a11) edge (a21);
    \draw (a21) .. controls +(270:1cm) and +(90:1cm) .. node [above]
    {\tiny $\{x_1\}$} (b2);
    \draw (b2) edge (postdots_b2);
    \draw (b2) edge node [above] {\tiny $\begin{array}{r}y\leq d\\
        \wedge\,z\geq 1\end{array}$} node [below] {\tiny $\{z\}$}
    (a12); 
    \draw (a12) edge (a22);
    \draw (a22) .. controls +(270:1cm) and +(90:1cm) .. node [above]
    {\tiny $\{x_1\}$} (b3);
    \draw (b3) edge (postdots_b3);
    \draw (b3) edge node [above] {\tiny $\begin{array}{r}y\leq d\\
        \wedge\,z\geq 1\end{array}$} node [below] {\tiny $\{z\}$} (a13);
    \draw[dotted] (a13) edge (a23);
    \draw (a23) .. controls +(270:1cm) and +(90:1cm) .. node
    [above] {\tiny $\{x_1\}$} (b4);
    \draw (b4) edge (postdots_b4);
    \draw (b4) edge node [above] {\tiny $\begin{array}{r}y\leq d\\
        \wedge\,z\geq 1\end{array}$} node [below] {\tiny $\{z\}$}
    (a14);
    \draw (a14) edge (a24);
    \draw (a24) .. controls +(270:1cm) and +(90:1cm) .. node [above]
    {\tiny $\{x_1\}$} (b5);
    \draw (b5) edge (postdots_b5);
  \end{tikzpicture}
  \caption{Part of $ZG(\Ww_2)$.}
  \label{fig:zgw}
\end{figure}

We now exploit this fact to give an example of a TBA $\Aa_n$ whose
zone graph has a number of zones linear in the number of clocks, but
$\Bb_n=SNZ(\Aa_n)$ has a zone graph of size exponential in the number
of clocks.

\begin{figure}
  \begin{tikzpicture}[shorten >=1pt,node distance=1cm,on grid]
    \node[draw,initial,initial text=] (Rn) {\footnotesize $\Rr_n$};
    \node[draw] (Vn) [right of=Rn] {\footnotesize $\Vv_n$};
    \node (middledots) [right of=Vn] {\footnotesize $\dots$};
    \node[draw] (R2) [right of=middledots] {\footnotesize $\Rr_2$};
    \node[draw] (V2) [right of=R2] {\footnotesize $\Vv_2$};
    \draw[->] (Rn) edge (Vn);
    \draw[->] (Vn) edge (middledots);
    \draw[->] (middledots) edge (R2);
    \draw[->] (R2) edge (V2);
    \node (An) [below of=middledots,node distance=0.7cm]
    {$\mathbf{\Aa_n}$};
  \end{tikzpicture}
  \hfill
  \begin{tikzpicture}[shorten >=1pt,node distance=1cm,on grid]
    \node[draw,initial,initial text=] (Rn) {\footnotesize $\Rr_n$};
    \node[draw] (Wn) [right of=Rn] {\footnotesize $\Ww_n$};
    \node (middledots) [right of=Wn] {\footnotesize $\dots$};
    \node[draw] (R2) [right of=middledots] {\footnotesize $\Rr_2$};
    \node[draw] (W2) [right of=R2] {\footnotesize $\Ww_2$};
    \draw[->] (Rn) edge (Wn);
    \draw[->] (Wn) edge (middledots);
    \draw[->] (middledots) edge (R2);
    \draw[->] (R2) edge (W2);
    \node (Bn) [below of=middledots,node distance=0.7cm]
    {$\mathbf{\Bb_n}$};
  \end{tikzpicture}
  \caption{Automata $\Aa_n$ (left) and $\Bb_n=SNZ(\Aa_n)$ (right).}
  \label{fig:An-Bn}
\end{figure}

$\Aa_n$, in Figure~\ref{fig:An-Bn}, is constructed from the automata
gadgets $\Vv_k$ and $\Rr_k$ as shown in Figures~\ref{fig:gadget-Vk-Wk}
and~\ref{fig:gadget-Rk}. Observe that the role of $\Rr_k$ is to
enforce an order $0\leq y\leq x_1 \leq \dots\leq x_k$ between clock
values.  By induction on $k$ one can compute that there are only two
zones at locations $b^k$ since $\Rr_{k+1}$ made the two zones in
$b^{k+1}$ collapse into the same zone in $b^k$. Hence the number of
nodes in the zone graph of $\Aa_n$ is $\Oo(n)$.

Let us now consider $\Bb_n$, the strongly non-Zeno automaton obtained
from $\Aa_n$ following~\cite{TYB:FMSD:2005}. Every gadget $\Vv_k$ gets
transformed to $\Ww_k$ as shown in Figure~\ref{fig:An-Bn}. While
exploring $\Ww_k$, one introduces a distance between the clocks
$x_{k-1}$ and $x_k$. So when leaving it one gets zones with $x_{k} -
x_{k-1} \ge c$, where $c \in \{0,1,2, \ldots , d \}$. The distance
between $x_k$ and $x_{k-1}$ is preserved by $\Rr_k$. In consequence,
$\Ww_n$ produces at least $d+1$ zones. For each of these zones
$\Ww_{n-1}$ produces $d+1$ more zones. In the end, the zone graph of
$\Bb_n$ has at least $(d+1)^{n-1}$ zones at the state $b^2$. The zones
obtained with the state $b^k$ are of the form

\begin{align*}
  & 0 = x_1 = \ldots = x_{k-1} \le z \le y \le x_{k} \le
  \ldots \le x_n\\
  & \wedge\ \bigwedge_{i\in\{k, \dots, n-1\}} x_{i+1} - x_{i} \ge c_i
  \qquad\qquad\qquad\qquad\qquad
  \text{where each}\ c_i \in \{0, 1, \dots, d\}
\end{align*}

So the zone graph has at least $(d+1)^{n-k+1}$ zones at state
$b^k$. Hence, the zone graph of $\Bb_n$ contains at least
$(d+1)^{n-1}$ zones.

We have thus shown that $\Aa_n$ has $\Oo(n)$ zones while
$\Bb_n=SNZ(\Aa_n)$ has an exponential number of zones even when the
constant $d$ is $1$. One could argue that the transformation
in~\cite{TYB:FMSD:2005} can be transformed in such a way to prevent
the combinatorial explosion. In particular, it is often suggested to
replace $z\geq 1$ by a guard that matches the biggest constant in the
automaton, that is $z\geq d$ in our case. However, this would still
yield an exponential blowup as every zone with state $b^k$ yields two
different zones with state $b^{k-1}$ that do not collapse going
through $\Rr_{k-1}$. Observe also that the construction shows that
even with two clocks the number of zones blows exponentially in the
binary representation of $d$. Note that the automaton $\Aa_n$ does not
have a non-Zeno accepting run. Hence, every search algorithm is
compelled to explore all the zones of $\Bb_n$.

\subsection{A more efficient solution}
\label{sec:non_zeno:more_efficient_solution}

We aim to decide if a given path in a zone graph has a non-Zeno
instantiation. This is equivalent to deciding if all instantiations
of a path are Zeno. There are essentially two reasons for this:
\begin{itemize}
\item there may be a clock $x$ that is reset finitely many times but
  bound infinitely many times by guards  $x\le c$:
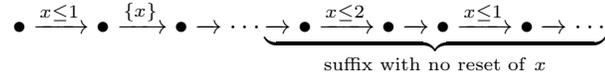
  \begin{equation*}
    \bullet \xra{x\le 1} \bullet \xra{\{x\}} \bullet \xra{} \cdots
    \underbrace{\xra{} \bullet \xra{x\leq 2} \bullet \xra{} \bullet
    \xra{x\le 1} \bullet \xra{} \cdots}_{\text{suffix with no reset
      of}\ x}
  \end{equation*}
  
\item or time may not be able to elapse at all due to infinitely many
  transitions that check $x=0$, forcing $x$ to stay at $0$:
  \begin{equation*}
    \bullet \xra{\{y\}} \bullet \xra{x=0} \bullet \xra{\{x\}} \bullet
    \xra{y=0} \bullet \xra{\{y\}} \bullet \xra{x=0} \cdots
  \end{equation*}
\end{itemize}

Our solution stems from a realization that we only need one non-Zeno
run satisfying the Büchi condition and so in a way transforming an
automaton to strongly non-Zeno is excessive. We propose not to modify
the automaton, but to introduce additional information to the zone
graph $ZG^a(\Aa)$. The nodes will now be triples $(q,Z,Y)$ where
$Y\incl X$ is the set of clocks that can potentially be equal to
$0$. It means in particular that other clock variables, i.e. those
from $X-Y$ are assumed to be bigger than $0$. We write $(X-Y)>0$ for
the constraint saying that all the variables in $X-Y$ are not $0$.

\begin{definition}\label{defn:zero_guessing_zone_graph}
  Let $\Aa$ be a TBA over a set of clocks $X$.  The \emph{guessing
    zone graph} $GZG^a(\mathcal{A})$ has nodes of the form $(q,Z,Y)$
  where $(q,Z)$ is a node in $ZG^a(\Aa)$ and $Y\incl X$. The initial
  node is $(q_0,Z_0,X)$, with $(q_0,Z_0)$ the initial node of
  $ZG^a(\Aa)$. In $GZG^a(\Aa)$ there are transitions:
  \begin{itemize}
  \item $(q,Z,Y)\xra{t}(q',Z',Y\cup R)$ if there is a
    transition $(q,Z)\xra{t}(q',Z')$ in $ZG^a(\Aa)$ with
    $t=(q,g,R,q')$, and there are valuations $\val\in Z$, $\val'\in
    Z'$, and $\d\in\Rpos$ such that $\val+\d\sat (X-Y)>0$ and
    $(q,\val)\xra{\d,t}(q,\val')$;
  \item $(q,Z,Y)\xra{\t}(q,Z,Y')$,   on a  new auxiliary letter $\t$,
    for $Y'=\es$ or $Y'=Y$.
  \end{itemize}
\end{definition}
The additional component $Y$ expresses some information about possible
valuations with which we can take a transition. The first case is about
transitions that are realizable when clocks outside $Y$ are
positive. While it is formulated in a more general way, one can think
of this transition as being instantaneous: $\d=0$. Then we have the second kind
of transitions, namely the transitions on $\t$, that allow us to
nondeterministically guess when time can pass.

It will be useful to distinguish some types of transitions and nodes
of $GZG^a(\Aa)$.

\begin{definition}
  We call a transition of $GZG^a(\Aa)$ a \emph{zero-check} when some
  clock is forced to be equal to $0$ by the guard $g$ of the
  transition; formally, for some clock $x$, for all $\val\in Z$, and
  all $\d\in\Rpos$ such that $\val+\d\sat g$ we have $(\val+\d)(x)=0$.
\end{definition}

The role of $Y$ sets will become obvious in the construction below. In
short, from a node $(q,Z,\emptyset)$, that is with $Y=\es$, every
reachable zero-check will be preceded by the reset of the variable
that is checked, and hence nothing prevents a time elapse in this
node. We will be particularly interested in the following types of
nodes to find non-Zeno accepting runs.

\begin{definition}
  A node $(q,Z,Y)$ of $GZG^a(\Aa)$ is \emph{clear} if the third
  component is empty: $Y=\es$. A node is A node is \emph{accepting} if
  $q$ is an accepting state.
\end{definition}

\begin{example}
  Figure~\ref{fig:A1_GZG} depicts a TBA $\Aa_1$ along with its zone
  graph $ZG^a(\Aa_1)$ and its guessing zone graph $GZG^a(\Aa_1)$ where
  $\tau$-loops have been omitted.

  The guessing zone graph construction can be optimized by restricting
  the guessed sets to clocks that are indeed equal to zero in some
  valuation in the zone. For instance, from the node $(b,x\geq
  1,\{x\})$ in Figure~\ref{fig:A1_GZG}, $x$ cannot be checked for zero
  unless it is first reset. Hence, this node can safely be removed
  from $\GZG(\Aa_1)$, yielding a smaller graph. In the resulting
  graph, the only loop goes through a $\tau$ transition. This
  emphasizes that time must elapse from node $(a,x=0,\{x\})$ in order
  to take a transition with guard $x\geq 1$. An optimized guessing
  zone graph construction is given in~\cite{Herbreteau:CONCUR:2011}.
\end{example}

\begin{figure}[ht]
  \footnotesize
  \begin{tikzpicture}[shorten >=1pt,node distance=2cm,on
    grid,auto,initial text=]
    \node[state,initial] (a) {$a$}; \node[state] (b) [below=of a]
    {$b$}; \draw[->] (a) edge [bend left] node {$x\geq 1$} (b);
    \draw[->] (b) edge node {$\begin{array}{l}x\leq 1\\
        \{x\}\end{array}$} (a);
    \node[node distance=1cm] (figname) [below=of b]
    {$\mathbf{\Aa_1}$};
  \end{tikzpicture}
  \hfill
  \begin{tikzpicture}[shorten >=1pt,node distance=2cm,on
    grid,auto,initial text=,every state/.style={rectangle,rounded
      corners=2mm}]
    \node[state,initial] (a) {$a,x=0$}; \node[state] (b) [below=of a]
    {$b,x\geq 1$}; \draw[->] (a) edge [bend left] node {$x\geq 1$}
    (b);
    \draw[->] (b) edge node {$\begin{array}{l}x\leq 1\\
        \{x\}\end{array}$} (a);
    \node[node distance=1cm] (figname) [below=of b]
    {$\mathbf{ZG^a(\Aa_1)}$};
  \end{tikzpicture}
  \hfill
  \begin{tikzpicture}[shorten >=1pt,node distance=2cm and 4cm,on
    grid,auto,every state/.style={rectangle,rounded
      corners=3mm},initial text=]
    \node[state,initial] (ay) {$a,x=0,\{x\}$}; \node[state] (ae)
    [right=of ay] {$a,x=0,\emptyset$}; \node[state] (by) [below=of ay]
    {$b,x\geq 1,\{x\}$}; \node[state] (be) [right=of by] {$b,x\geq
      1,\emptyset$}; \draw[->] (ay) edge node {$x\geq 1$} (by);
    \draw[->,dashed] (ay) edge node {$\tau$} (ae); \draw[->] (ae) edge
    node {$x\geq 1$} (be);
    \draw[->] (by) edge [bend left] node {$\begin{array}{l} x\leq 1\\
        \{x\}\end{array}$} (ay); \draw[->,dashed] (by) edge node
    {$\tau$} (be); \draw[->] (be) edge node [pos=0.35,sloped]
    {$\begin{array}{l} x\leq 1\\ \{x\}\end{array}$} (ay); \node[node
    distance=1cm and 2cm] (figname) [below right=of by]
    {$\mathbf{GZG^a(\Aa_1)}$};
  \end{tikzpicture}
  \caption{A TBA $\Aa_1$ and the guessing zone graph $GZG^a(\Aa_1)$
    (with $\tau$ self-loops omitted for clarity).}
  \label{fig:A1_GZG}
\end{figure}

Notice that directly from the definition it follows that a path in
$GZG^a(\Aa)$ determines a path in $ZG^a(\Aa)$ obtained by removing
$\t$ transitions and the third component from nodes.

In order to state the main theorem succinctly we need some notions

\begin{definition}
  A variable $x$ is \emph{bounded} by a transition of $GZG^a(\Aa)$ if
  the guard of the transition implies $x\leq c$ for some constant
  $c$. More precisely: $x$ is bounded by the transition
  $(q,Z,Y)\xra{(q,g,R,q')}(q',Z',Y')$, if for all $\val\in Z$ and
  $\d\in\Rpos$ such that $\val+\d\sat g$, we have $(\val+\d)(x)\leq c$
  for some $c\in\Nat$. A variable is \emph{reset} by the transition if
  it belongs to the reset set $R$ of the transition.
\end{definition}

\begin{definition}
  \label{defn:blocked_path}
  We say that a path is \emph{blocked} if there is a variable that is
  bounded infinitely often and reset only finitely often by the
  transitions on the path. Otherwise the path is called
  \emph{unblocked}. 
\end{definition}
Obviously, paths corresponding to non-Zeno runs are
unblocked.

\begin{theorem}
  \label{thm:zeno_on_zero_guessing}
  A TBA $\mathcal{A}$ has a non-Zeno run satisfying the Büchi
  condition iff there exists an unblocked path in $GZG^a(\Aa)$
  visiting both an accepting node and a clear node infinitely often.
\end{theorem}

The proof of Theorem~\ref{thm:zeno_on_zero_guessing} follows from
Lemmas~\ref{lem:TBA_to_zero_guessing}
and~\ref{lem:zero_guessing_to_TBA} below. It is in
Lemma~\ref{lem:zero_guessing_to_TBA} that the third component of
states is used.

At the beginning of the section we had recalled that the progress
criterion~\cite{AD:TCS:1994} stated in page~\pageref{progressive}
characterizes the paths in region graphs that have non-Zeno
instantiations. We had mentioned that it cannot be directly extended
to zone graphs since their transitions are not
pre-stable. Lemma~\ref{lem:zero_guessing_to_TBA} below shows that by
slightly complicating the zone graph we can recover a result very
similar to Lemma~4.13 in~\cite{AD:TCS:1994}.

\begin{lemma}
  \label{lem:TBA_to_zero_guessing}
  If $\mathcal{A}$ has a non-Zeno run satisfying the Büchi condition,
  then in $GZG^a(\Aa)$ there is an unblocked path visiting both an
  accepting node and a clear node infinitely often.
\end{lemma}

\begin{proof}
  Let $\r$ be a non-Zeno run of $\mathcal{A}$:
  \begin{equation*}
    (q_0,\nu_0)\xrightarrow{\delta_0,t_0}
    (q_1,\nu_1)\xrightarrow{\delta_1,t_1}
    \cdots
  \end{equation*}
  By Theorem~\ref{thm:concretization_tri09}, it is a concretization of
  a path $\s$ in $ZG^a(\mathcal{A})$:
  \begin{equation*}
    (q_0,Z_0)\xrightarrow{t_0}
    (q_1,Z_1)\xrightarrow{t_1}
    \cdots
  \end{equation*}

  Let $\sigma'$ be the following sequence:
  \begin{equation*}
    (q_0,Z_0,Y_0)\xrightarrow{\tau}
    (q_0,Z_0,Y_0')\xrightarrow{t_0}
    (q_1,Z_1,Y_1)\xrightarrow{\tau}
    (q_1,Z_1,Y_1')\xrightarrow{t_1}
    \cdots
  \end{equation*}
  where $Y_0=X$, $Y_i$ is determined by the transition, and $Y'_i=Y_i$
  unless $\d_i>0$ when we put $Y'_i=\es$. We need to see that this is
  indeed a path in $GZG^a(\Aa)$. For this we need to see that every
  transition
  $(q_i,Z_i,Y_i')\xrightarrow{t_i}(q_{i+1},Z_{i+1},Y_{i+1})$ is
  realizable from a valuation $\val$ such that $\val\sat
  (X-Y'_i)>0$. But an easy induction on $i$ shows that actually
  $\val_i\sat (X-Y'_i)>0$.

  Since $\rho$ is non-Zeno there are infinitely many $i$ with
  $Y'_i=\es$. Since the initial run is non-Zeno, $\s'$ is
  unblocked.
  \qed
\end{proof}

\begin{lemma}
  \label{lem:zero_guessing_to_TBA}
  Suppose $GZG^a(\mathcal{A})$ has an unblocked path visiting
  infinitely often both a clear node and an accepting node then
  $\mathcal{A}$ has a non-Zeno run satisfying the Büchi condition.
\end{lemma}

\begin{proof}
  Let $\s$ be a path in $GZG^a(\Aa)$ as required by the assumptions of
  the lemma (without loss of generality we assume every alternate
  transition is a $\t$ transition):
  \begin{equation*}
    (q_0,Z_0,Y_0)\xrightarrow{\tau} (q_0,Z_0,Y_0')\xrightarrow{t_0}
    \cdots
    (q_i,Z_i,Y_i)\xrightarrow{\tau} (q_i,Z_i,Y_i')\xrightarrow{t_i}
    \cdots
  \end{equation*}
  Take a corresponding path in $ZG^a(\Aa)$ and one instantiation
  $\r=(q_0,\nu_0),(q_1,\nu_1)\dots$ that exists by
  Theorem~\ref{thm:concretization_tri09}. If it is non-Zeno then we
  are done.

  Suppose $\r$ is Zeno. We now show how to build a non-Zeno
  instantiation of $\s$ from $\r$. Let $X^r$ be the set of variables
  reset infinitely often on $\s$. As $\s$ is unblocked, every variable
  not in $X^r$ is bounded only finitely often. Since $\r$ is Zeno,
  there is an index $m$ such that the duration of the suffix of the
  run starting from $(q_m,\nu_m)$ is bounded by $1/2$, and no
  transition in this suffix bounds a variable outside $X^r$. Let $n>m$
  be such that every variable from $X^r$ is reset between $m$ and
  $n$. Observe that $\nu_n(x)< 1/2$ for every $x\in X^r$.

  Take positions $i,j$ such that $i,j>n$, $Y_i=Y_j=\es$ and all the
  variables from $X^r$ are reset between $i$ and $j$. We look at the
  part of the run $\r$:
  \begin{equation*}
    (q_i,\val_i)\xra{\d_i,t_i}
    (q_{i+1},\val_{i+1})\xra{\d_{i+1},t_{i+1}} \dots(q_j,\val_j)
  \end{equation*}
  and claim that for every $\zeta\in \Rpos$ the sequence of the form
  \begin{equation*}
    (q_i,\val'_i)\xra{\d_i,t_i}
    (q_{i+1},\val'_{i+1})\xra{\d_{i+1},t_{i+1}} \dots(q_j,\val'_j)
  \end{equation*}
  is a part of a run of $\Aa$ where $\val'_k$ for $k=i,\dots,j$ satisfy:
  \begin{enumerate}
  \item $\nu'_k(x)= \nu_k(x)+\zeta+1/2$ for all $x\not\in X^r$,

  \item $\nu'_k(x)=\nu_k(x)+1/2$ if $x\in X^r$ and $x$ has not been
    reset between $i$ and $k$.
  \item $\nu'_k(x)=\nu_k(x)$ otherwise, i.e., when $x\in X^r$ and $x$
    has been reset between $i$ and $k$.
  \end{enumerate}
  Before proving this claim, let us explain how to use it to conclude
  the proof.  The claim shows that in $(q_i,\val_i)$ we can pass $1/2$
  units of time and then construct a part of the run of $\Aa$ arriving
  at $(q_j,\val'_j)$ where $\val'_j(x)=\val_j(x)$ for all variables in
  $X^r$, and $\val'_j(x)=\val_j(x)+1/2$ for other variables.  Now, we
  can find $l>j$, so that the pair $(j,l)$ has the same properties as
  $(i,j)$. We can pass $1/2$ units of time in $j$ and repeat the above
  construction getting a longer run that has passed $1/2$ units of
  time twice. This way we construct a run that passes $1/2$ units of
  time infinitely often, hence it is non-Zeno. By the construction it
  passes also infinitely often through accepting nodes.

  It remains to prove the claim. Take a transition
  $(q_k,\val_k)\xra{\d_k,t_k}(q_{k+1},\val_{k+1})$ and show that
  $(q_k,\val'_k)\xra{\d_k,t_k}(q_{k+1},\val'_{k+1})$ is also a
  transition allowed by the automaton. Let $g$ and $R$ be the guard of
  $t_k$ and the reset of $t_k$, respectively.

  First we need to show that $\val'_k+\d_k$ satisfies the guard of
  $t_k$. For this, we need to check if for every variable $x\in X$ the
  constraints in $g$ concerning $x$ are satisfied. We have three
  cases:
  \begin{itemize}
  \item If $x\not\in X^r$ then $x$ is not bounded by the transition
    $t_k$, that means that in $g$ the constraints on $x$ are of the
    form $(x>c)$ or $(x\geq c)$. Since $(\val_k+\d_k)(x)$ satisfies
    these constraints so does $(\val_k'+\d_k)(x)\geq
    (\val_k+\d_k)(x)$.
  \item If $x\in X^r$ and it is reset between $i$ and $k$ then
    $\val_k'(x)=\val_k(x)$ so we are done.
  \item Otherwise, we observe that $x\not\in Y_k$. This is because
    $Y_i=\es$, and then only variables that are reset are added to
    $Y$. Since $x$ is not reset between $i$ and $k$, it cannot be in
    $Y_k$. By definition of transitions in $GZG^a(\Aa)$ this means
    that $g\land (x>0)$ is consistent. We have that $0\leq
    (\val_k+\d_k)(x)< 1/2$ and $1/2\leq (\val_k'+\d_k)(x)< 1$. So
    $\val_k'+\d_k$ satisfies all the constraints in $g$ concerning $x$
    as $\val_k+\d_k$ does.
  \end{itemize}

  This shows that there is a transition
  $(q_k,\val'_k)\xra{\d_k,t_k}(q_{k+1},\val')$ for the uniquely
  determined $\val'=\reset{R}(\val_k'+\d_k)$. It is enough to show
  that $\val'=\val'_{k+1}$. For variables not in $X^r$ it is clear as
  they are not reset. For variables that have been reset between $i$
  and $k$ this is also clear as they have the same values in
  $\val_{k+1}'$ and $\val'$. For the remaining variables, if a
  variable is not reset by the transition $t_k$ then its value is the
  same in $\val'$ and $\val_k'$. If it is reset then its value in
  $\val'$ becomes $0$; but so it is in $\val_{k+1}'$, and so the third
  condition holds. This proves the claim.  \qed
\end{proof}

Finally, we provide an explanation as to why the proposed solution
does not produce an exponential blowup. At first it may seem that we
have gained nothing because when adding arbitrary sets $Y$ we have
automatically caused exponential blowup to the zone graph. We claim
that this is not the case for the part of $GZG^a(\Aa)$ reachable from
the initial node, namely a node with the initial state of $\Aa$,
the zone putting every clock to $0$, and $Y=X$.

We say that a zone \emph{orders clocks} if for every two clocks $x,
y$, the zone implies that at least one of $x\leq y$, or $y\leq x$
holds.

\begin{lemma}\label{lemma:zones order clocks}
  If a node with a zone $Z$ is reachable from the initial node of the
  zone graph $ZG^a(\Aa)$ then $Z$ orders clocks. The same holds for
  $GZG^a(\Aa)$.
\end{lemma}

\begin{proof}
  First notice that in the initial zone, all the clocks are equal to
  each other. Now, consider a zone $Z$ that orders clocks. Let
  $(q,Z)\xra{t}(q',Z')$ be a transition of $ZG^a(\Aa)$. This means
  that there exists a transition $(q,Z)\xra{t} (q',Z_1')$ in the
  (unabstracted) zone graph $ZG(\Aa)$ such that $Z'=\Approx_M(Z_1')$.
  Directly from the definition of transitions we have that $Z'_1$
  orders clocks.  It remains to check that, the clock ordering in
  $Z'_1$ is preserved in $Z'=\Approx_M(Z'_1)$. Suppose not, then let
  $x_1\le \dots\le x_n$ be the ordering in $Z_1'$. We get that
  $Z'\wedge (x_1\le \dots\le x_n)$ is a smaller convex union of
  d-regions than $Z'$ that contains $Z_1'$ (recall that $M \ge 0$) --
  a contradiction. For the second statement observe that for every
  node $(q,Z,Y)$ in $GZG^a(\Aa)$, $(q,Z)$ is reachable in $ZG^a(\Aa)$.
  \qed
\end{proof}

Suppose that $Z$ orders clocks. We say that a set of clocks $Y$
\emph{respects the order given by $Z$} if whenever $y\in Y$ and $Z$
implies $x\leq y$ then $x\in Y$. In other words, $Y$ is downward
closed with respect to the ordering constraint in $Z$.

\begin{lemma}\label{lemma:Y respects order}
  If a node $(q,Z,Y)$ is reachable from the initial node of the
  guessing zone graph $GZG^a(\Aa)$ then $Y$ respects the order given
  by $Z$.
\end{lemma}

\begin{proof}
  The proof is by induction on the length of a path.  In the initial
  node $(q_0,Z_0,X)$, the set $X$ obviously respects the order as it
  is the set of all clocks. Now take a transition
  $(q,Z,Y)\xra{t}(q',Z',Y')$ with $Y$ respecting the order in $Z$. We
  need to show that $Y'$ respects the order in $Z'$. By the definition
  of transitions in $GZG^a(\Aa)$ there are $\val\in Z$, $\val'\in Z'$
  and $\d\in\Rpos$ such that $(q,\val)\xra{\d,t}(q',\val')$ and
  $\val+\d\sat(X-Y)>0$. Take $y\in Y'$ and suppose that $Z'$ implies
  $x\leq y$ for some clock $x$. There are three cases depending on
  which of the variables $y$, $x$ are being reset by the transition.

  \begin{itemize}
  \item If $x$ is reset by the transition then, by definition $x\in
    Y'$.
  \item If $y$ is reset then $Z'$ implies $y=0$. Hence $Z'$ implies
    that $x=0$. When $x$ is not reset, $x$ is checked for $0$ on
    $t$. Hence, $x\in Y$ and $x\in Y'$.
  \item The remaining case is when none of the two variables is reset by
    the transition. As $\val'\in Z'$, we have that $\val'\sat x\leq
    y$; and in consequence $\val\sat x\leq y$. Since $Z$ orders clocks
    and $\val\in Z$, we must have that $Z$ implies $x\leq y$.  As $y$
    has not been reset, $y\in Y$. By assumption that $Y$ orders
    clocks, $x\in Y$.
  \end{itemize}
  \qed
\end{proof}

The above two lemmas give us the desired bound.

\begin{theorem}\label{thm:GZG size}
  Let $|ZG^a(\Aa)|$ be the size of the zone graph, and $|X|$ be the
  number of clocks in $\Aa$.  The number of reachable nodes of
  $GZG^a(\Aa)$ is bounded by $|ZG^a(\Aa)|.(|X|+1)$.
\end{theorem}
The theorem follows directly from the above two lemmas. Of course,
imposing that zones have ordered clocks in the definition of
$GZG^a(\Aa)$ we would get the same bound for the entire $GZG^a(\Aa)$.

\subsection{Examples of guessing zone graphs}

Figure~\ref{fig:A1_GZG} in Section
\ref{sec:non_zeno:more_efficient_solution} depicts a TBA $\Aa_1$ along
with $ZG^a(\Aa_1)$ and $GZG^a(\Aa_1)$ (where the $\tau$-loops have
been omitted). In order to fire transition $b\xra{x\leq 1,\{x\}} a$
time must not elapse in $b$. The third component $Y$ does not help to
detect that time cannot elapse in $b$ as in $GZG^a(\Aa_1)$ the
transition is allowed for both $Y=\{x\}$ and $Y=\emptyset$. However,
as soon as a strongly-connected component (SCC) contains a transition
$x\geq 1$ and a transition that resets $x$, it has a non-Zeno run, and
the third component does not play any role.

The third component is only useful for the case where an SCC contains
no transition with a guard implying $x>0$ for some clock $x$ that is
also reset on some transition in the SCC. In such a case, zero-checks
may prevent time to elapse. We illustrate this case on the next two
examples that emphasize how the third component added to the states of
the zone graph allows to distinguish between Zeno runs and non-Zeno
runs.

\begin{figure}[p]
  \centering
  \begin{tikzpicture}[shorten >=1pt,node distance=2cm,on
    grid,auto,initial text=,initial where=above]
    \node[state,initial] (0) {$0$};
    \node[state] (1) [left=of 0] {$1$};
    \node[state] (2) [right=of 0] {$2$};
    \draw[->] (0) edge [bend left] node {$\{x\}$} (1);
    \draw[->] (0) edge [bend left] node {$\{y\}$} (2);
    \draw[->] (1) edge [bend left] node {$y=0$} (0);
    \draw[->] (2) edge [bend left] node {$x=0$} (0);
    \node[node distance=1.5cm] (figname) [left=of 1]
    {$\mathbf{\Aa_2}$};
  \end{tikzpicture}

  \vspace{0.5cm}
  \begin{tikzpicture}[initial text=,initial where=right]
    \begin{scope}
      \tikzstyle{every node}=[draw,rounded corners=2mm]
      \node (z2) at (0,0) {{\footnotesize $z2:(1,0=x\leq
          y),\emptyset$}};
      \node (z2x) at (4,0) {{\footnotesize $z2,\{x\}$}};
      \node (z2xy) at (8,0) {{\footnotesize $z2,\{x,y\}$}};
      \node (z1) at (0,-1.5) {{\footnotesize
          $z1:(0,0=x=y),\emptyset$}};
      \node (z1xy) at (8,-1.5) {{\footnotesize $z1,\{x,y\}$}};
      \node (z3) at (0,-3) {{\footnotesize $z3:(2,0=y\leq
          x),\emptyset$}};
      \node (z3y) at (4,-3) {{\footnotesize $z3,\{y\}$}};
      \node (z3xy) at (8,-3) {{\footnotesize $z3,\{x,y\}$}};
    \end{scope}
    \begin{scope}[->,line width=1pt]
      \draw (9,-1.5) -- (z1xy);
      \draw (z1) -- node[above] {{\footnotesize $\{x\}$}} (z2x);
      \draw (z1) -- node[below] {{\footnotesize $\{y\}$}} (z3y);
      \draw (z1xy) edge [bend left] node [left] {{\footnotesize
          $\{x\}$}} (z2xy);
      \draw (z1xy) edge [bend left] node [right] {{\footnotesize
          $\{y\}$}} (z3xy);
      \draw (z2xy) edge [bend left] node [right] {{\footnotesize
          $y=0$}} (z1xy);
      \draw (z3xy) edge [bend left] node [left] {{\footnotesize
          $x=0$}} (z1xy);
      \begin{scope}[dashed]
        \draw (z1xy) -- node[above] {{\footnotesize $\tau$}} (z1);
        \draw (z2x) -- node[above] {{\footnotesize $\tau$}} (z2);
        \draw (z2xy) .. controls +(135:1cm) and +(45:1cm)
        .. node[above] {{\footnotesize $\tau$}} (z2);
        \draw (z3y) -- node[below] {{\footnotesize $\tau$}} (z3);
        \draw (z3xy) .. controls +(225:1cm) and +(315:1cm)
        .. node[below] {{\footnotesize $\tau$}} (z3);
      \end{scope}
    \end{scope}
    \node[node distance=1cm] (figname) [below=of z3y]
    {\textbf{Reachable part of} $\mathbf{GZG^a(\Aa_2)}$};
  \end{tikzpicture}

  \vspace{1cm}

  \begin{tikzpicture}[shorten >=1pt,node distance=2cm,on
    grid,auto,initial text=,initial where=above]
    \node[state,initial] (0) {$0$};
    \node[state] (1) [left=of 0] {$1$};
    \node[state] (2) [right=of 0] {$2$};
    \draw[->] (0) edge [bend left] node {$\{x\}$} (1);
    \draw[->] (0) edge [bend left] node {$x=0$} (2);
    \draw[->] (1) edge [bend left] node {$y=0$} (0);
    \draw[->] (2) edge [bend left] node {$\{y\}$} (0);
    \node[node distance=1.5cm] (figname) [left=of 1] {$\mathbf{\Aa_3}$};
  \end{tikzpicture}

  \vspace{0.5cm}
  \begin{tikzpicture}[initial text=,initial where=right]
    \begin{scope}
      \tikzstyle{every node}=[draw,fill=white,rounded corners=2mm]
      \node (z2) at (0,0) {{\footnotesize $z2:(2,0=x=y),\emptyset$}};
      \node (z2xy) at (8,0) {{\footnotesize $z2,\{x,y\}$}};
      \node (z3) at (0,-1.7) {{\footnotesize $z3:(0,0=y\leq
          x),\emptyset$}};
      \node (z3y) at (4,-1.7) {{\footnotesize $z3,\{y\}$}};
      \node (z3xy) at (8,-1.7) {{\footnotesize $z3,\{x,y\}$}};
      \node (z4) at (0,-3.4) {{\footnotesize $z4:(1,0=x\leq
          y),\emptyset$}};
      \node (z4x) at (4,-3.4) {{\footnotesize $z4,\{x\}$}};
      \node (z4xy) at (8,-3.4) {{\footnotesize $z4,\{x,y\}$}};
      \node (z1) at (0,-5.1) {{\footnotesize $z1:(0,0=x=y),\emptyset$}};
      \node (z1xy) at (8,-5.1) {{\footnotesize $z1,\{x,y\}$}};
    \end{scope}
    \begin{scope}[->,line width=1pt]
      \draw (9,-5.1) -- (z1xy);
      \draw (z1) -- node[above] {{\footnotesize $\{x\}$}} (z4x);
      \draw (z1xy) .. controls +(15:1.5cm) and +(345:1.5cm)
      .. node[right] {{\footnotesize $x=0$}} (z2xy);
      \draw (z1xy) -- node[right] {{\footnotesize $\{x\}$}} (z4xy);
      \draw (z2) -- node[above] {{\footnotesize $\{y\}$}} (z3y);
      \draw (z2xy) .. controls +(250:0.8cm) and +(110:0.8cm)
      .. node[left] {{\footnotesize $\{y\}$}} (z3xy);
      \draw (z3) -- node[above] {{\footnotesize $\{x\}$}} (z4x);
      \draw (z3y) -- node[above] {{\footnotesize $\{x\}$}} (z4xy);
      \draw (z3xy) -- node[right] {{\footnotesize $x=0$}} (z2xy);
      \draw (z3xy) -- node[left] {{\footnotesize $\{x\}$}} (z4xy);
      \draw (z4xy) .. controls +(250:0.8cm) and +(110:0.8cm)
      .. node[left] {{\footnotesize $y=0$}} (z1xy);
      \begin{scope}[dashed]
        \draw (z1xy) -- node[below] {{\footnotesize $\tau$}} (z1);
        \draw (z2xy) -- node[above] {{\footnotesize $\tau$}} (z2);
        \draw (z3y) -- node[above] {{\footnotesize $\tau$}} (z3);
        \draw (z3xy) .. controls +(135:1cm) and +(45:1cm) .. node[above]
        {{\footnotesize $\tau$}} (z3);
        \draw (z4x) -- node[above] {{\footnotesize $\tau$}} (z4);
        \draw (z4xy) .. controls +(135:1cm) and +(45:1cm) .. node[above]
        {{\footnotesize $\tau$}} (z4);
      \end{scope}
    \end{scope}
    \node[node distance=2cm] (figname) [below=of z4x] {\textbf{Reachable
        part of} $\mathbf{GZG^a(\Aa_3)}$};
  \end{tikzpicture}
  \caption{Examples of guessing zone graphs ($\tau$ self-loops have been
    omitted for clarity)}
  \label{fig:gzg-eg}
\end{figure}

The TBA $\Aa_2$ shown in Figure \ref{fig:gzg-eg} has only runs where
the time cannot elapse at all. This is detected in $GZG^a(\Aa_2)$ as all
states in the only non-trivial SCC have $Y=\{x,y\}$ as the third
component. This means that from every state there exists a reachable
zero-check that is not preceded by the corresponding reset, hence
preventing time to elapse. Notice that the correctness of this
argument relies on the fact that for every $(q,Z,Y)$ in
$GZG^a(\Aa_2)$, and for every transition $t=(q,g,R,q')$, even if $t$
is fireable in $ZG^a(\Aa_2)$ from $(q,Z)$, it must also be fireable
\emph{under the supplementary hypothesis} $(X-Y)>0$ given by $Y$ in
$GZG^a(\Aa_2)$.

The TBA $\Aa_3$ in Figure \ref{fig:gzg-eg} admits a non-Zeno run. This
can be read from $GZG^a(\Aa_3)$ since the SCC composed of the four
zones with $Y=\{x,y\}$ together with $(z_2,\emptyset)$ and
$(z_3,\{y\})$ contains a clear node. This is precisely the state where
time can elapse as every reachable zero-check is preceded by the
corresponding reset.


\section{Algorithm}
\label{sec:algo}

In this section, we provide an on-the-fly algorithm for the B\"{u}chi
non-emptiness problem using the guessing zone graph construction
developed in Section~\ref{sec:non_zeno:more_efficient_solution}. In
the later part of the section, we observe that in most cases,
non-Zenoness could be detected directly from the standard zone graph,
without extra construction. We provide an optimized on-the-fly
algorithm taking into account these observations.

We will use Theorem~\ref{thm:zeno_on_zero_guessing} to algorithmically
check if an automaton $\Aa$ has a non-Zeno run satisfying the B\"uchi
condition. The theorem requires to find an unblocked path in
$GZG^a(\Aa)$ visiting both an accepting node and a clear node
infinitely often. This problem is similar to that of testing for
emptiness of automata with generalized Büchi conditions as we need to
satisfy two infinitary conditions at the same time. The requirement of
a path being unblocked adds additional complexity to the problem. The
best algorithms for testing emptiness of automata with generalized
Büchi conditions are based on Tarjan's algorithm for strongly
connected components (SCC)~\cite{SE:TACAS:2005,GS:MEMICS:2009}. So
this is the way we take here. In particular, we adopt the variant
given by Couvreur~\cite{couvreur1999fly,Couvreur:SPIN:2005}.

In general, the verification problem for timed systems involves
checking if a network of timed automata $\Aa_1,\dots,\Aa_n$ satisfies
a given property $\phi$. Assuming that $\phi$ can be translated into a
(timed) Büchi automaton $\Aa_{\neg\phi}$, we reduce the verification
problem to the emptiness of a timed Büchi automaton $\Aa$ defined as a
product $\Aa_1\times \Aa_2\times \cdots\times \Aa_n\times \Aa_{\neg
  \phi}$ for some synchronization policy. Couvreur's algorithm is an
extension of Tarjan's algorithm for computing maximal SCCs in a
graph. One of its main features is that it stops as soon as a (non
necessarily maximal) SCC with an accepting state has been found. In
addition, it handles multiple accepting conditions efficiently. To
this regard, the algorithm computes the set of accepting conditions in
each SCC of $\Aa$. Initially, each state $s$ in $\Aa$ is considered as
a trivial SCC labelled with the accepting conditions of $s$. The
algorithm computes the states of $\Aa$ on-the-fly in a depth-first
search (DFS) manner starting from the initial state. During the
search, when a cycle is found, all the SCCs in the cycle are merged
into a bigger SCC $\G$ that inherits their accepting conditions. If
$\G$ contains all the required accepting conditions, the algorithm
stops declaring $\Aa$ to be not empty. Notice that $\G$ need not be
maximal. Otherwise it resumes the DFS on $\Aa$. We direct the reader
to~\cite{couvreur1999fly,Couvreur:SPIN:2005,atva2010} for further
details on the Couvreur's algorithm.

In the next section, we show how to enhance Couvreur's algorithm
to detect runs that are not only accepting but also non-Zeno. It is
achieved by associating extra information to the SCCs in $\Aa$. This
information is updated when SCCs are merged like for accepting
conditions.

\subsection{Emptiness check on $\GZG(\Aa)$}
\label{sec:algo:gzg}

We apply Couvreur's algorithm for detecting maximal SCCs in
$GZG^a(\Aa)$. During the computation of the maximal SCCs, we keep
track of whether an accepting node and a clear node have been
seen. For the unblocked condition we use two sets of clocks $U_\G$ and
$R_\G$ that respectively contain the clocks that are bounded and the
clocks that are reset in the SCC $\G$. A clock from $U_\G-R_\G$ is
called \emph{blocking} since being bounded and not reset it puts a
limit on the time that can pass. At the end of the exploration of $\G$
we check if:
\begin{enumerate}
\item\label{condi} we have passed through an accepting node and a
  clear node,
\item\label{condii} there are no blocking clocks: $U_\G\subseteq
  R_\G$.
\end{enumerate}
If the two conditions are satisfied then we can conclude saying that
$\Aa$ has an accepting non-Zeno run. Indeed, a path passing infinitely
often through all the nodes of $\G$ would satisfy the conditions of
Theorem~\ref{thm:zeno_on_zero_guessing}, giving a required run of
$\Aa$. If the first condition does not hold then the same theorem says
that $\G$ does not have a witness for a non-Zeno run of $\Aa$
satisfying the Büchi condition.

The interesting case is when the first condition holds but not the
second. The following lemma yields an algorithm in that case.

\begin{lemma}
  Let $\G$ be an SCC in $\GZG(\Aa)$ with an accepting node and a clear
  node, and such that $U_\G\not\subseteq R_\G$. There exists an
  unblocked path in $\G$ that visits both an accepting node and a
  clear node infinitely often iff there exists a sub-SCC $\G'\subseteq
  \G$ with an accepting node and a clear node and such that
  $U_{\G'}\subseteq R_{\G'}$.
\end{lemma}

\begin{proof}
  Assume that $\G$ has an unblocked path that visits both an accepting
  node and a clear node infinitely often. Then, define $\G'$ as the
  set of nodes and edges that are visited infinitely often on that
  path.

  Conversely, if such a sub-SCC $\G'$ exists, then consider an
  infinite path in $\G'$ that goes infinitely often through each node
  and each transition in $\G'$. This path is unblocked and visits both
  an accepting node and a clear node. This path is also a path in
  $\G$.  \qed
\end{proof}

We call \emph{blocking edges} all the edges in $\G$ that bound a clock
from $U_\G\setminus R_\G$. We proceed as follows. We discard all the
blocking edges from $\G$ as every unblocked path in $\G$ goes only
finitely many times through these edges. In general, this yields
several candidates for $\G'$. Each of them is a proper sub-SCC of
$\G$. Then, we restart our algorithm on each such $\G'$. Since we have
discarded some edges from $\G$ (hence some resets), a clock may be now
blocking in $\G'$. If this is the case, the blocking edges in $\G'$
will be discarded, and the resulting sub-SCCs of $\G'$ will be
explored, and so on. Observe that each transition in $GZG^a(\Aa)$ will
be visited at most $|X|+1$ times, as we eliminate at least one clock
at each restart. If after exploring the entire graph, the algorithm
has not found a subgraph satisfying the two conditions then it
declares that there is no run of $\Aa$ with the desired
properties. The correctness of the procedure is based on
Theorem~\ref{thm:zeno_on_zero_guessing}. All the procedure: exploring
$\G$, discarding blocking edges, exploring all $\G'$ candidates, etc,
can be done on-the-fly without storing $\G$ as described
in~\cite{atva2010}.

Recall that by Theorem~\ref{thm:GZG size} the size of $GZG^a(\Aa)$ is
$\Oo(|ZG^a(\Aa)|\cdot|X|)$. The complexity of the algorithm follows
from the linear complexity of Couvreur's algorithm and the remark
about the bound on the number of times each transition is visited. We
hence obtain the following.

\begin{theorem}
  The above algorithm is correct and runs in time
  $\Oo(|ZG^a(\Aa)|\cdot |X|^2)$.
\end{theorem}

Although the guessing zone graph provides a way to detect non-Zeno
paths, it is useful only when the automaton indeed contains
zero-checks. The next challenge therefore lies in optimizing the use
of the guessing zone graph construction, that is, applying 
Couvreur's algorithm directly on the standard zone graph and using the
guessing zone graph construction only when required.

\subsection{Optimized use of guessing zone graph construction}
\label{sec:algo:optimized}

The idea is to apply Couvreur's algorithm directly on $\ZG(\Aa)$ and
find an SCC with an accepting node. An SCC is said to be
\emph{unblocked} if it contains no blocking clock; recall that it is a
clock $x$ that is checked for a guard which implies $x \le c$ for a
constant $c$ and that is reset in no transition of the SCC. An SCC is said to be
\emph{strongly non-Zeno} if it contains a clock $x$ that is both reset on a
transition of the SCC, and checked in a guard which implies $x \ge 1$ in the
SCC.

Non-Zenoness can be ensured if the SCC satisfies one of the following
conditions:
\begin{itemize}
\item It is unblocked and free from zero-checks. A zero-check is
  detected for a transition $(q,Z)\xra{g,R}(q',Z')$ and some clock $x$
  when for each $\val\in Z$ and $\d\in\Rpos$ such that $\val+\d\sat
  g$, we have $(\val+\d)(x)=0$.
\item The SCC is strongly non-Zeno: there is a clock $x$ that is reset in the
SCC and one of the transitions in the SCC implies $x \ge 1$.
\end{itemize}

For the second condition, note that such a reachable SCC instantiates
into a path $\rho$ of $\Aa$ whose suffix corresponds to repeated
traversal of this SCC. Every traversal resets $x$ and checks for a
guard that implies $x \ge 1$. Therefore, at least $1$ time unit
elapses in each traversal, implying that $\rho$ is a non-Zeno
run. Notice that this relies on the same principle as the one used
in the Strongly Non-Zeno construction~\cite{TYB:FMSD:2005} (see
Section~\ref{sec:non_zeno:adding_one_clock}). However, in our case we
exploit the information from $\Aa$: we do not add any new clock. Our
algorithm will compute on the fly the set $L_\G$ of clocks $x$ such
that $x\ge 1$ is implied by some guard in $\G$. This is done in the
same way as for $U_\G$ in the previous subsection. Then, $\G$
satisfies the second condition above if $L_\G\cap R_\G$ is not empty.

The first condition is justified by the following lemma.

\begin{lemma}
  If $\ZG(\Aa)$ has an unblocked path that visits an accepting node
  infinitely often, and has only finitely many transitions with
  zero-checks, then $\Aa$ has a non-Zeno run satisfying the B\"{u}chi
  condition.
\end{lemma}

\begin{proof}
  Let $\s$ be the path in $\ZG(\Aa)$ as required by the assumptions of
  the lemma:

  \begin{equation*}
    (q_0, Z_0) \xrightarrow{t_0} \dots (q_i, Z_i) \xrightarrow{t_i} \dots 
  \end{equation*}

  Since zero-checks occur only finitely often in $\s$, we can find $j$
  such that the suffix $(q_j, Z_j) \xrightarrow{t_j} \dots $ of $\s$
  contains no zero-checks in its transitions. Let $\s'$ be the
  following sequence:

  \begin{equation*}
    (q_0,Z_0,Y_0)\xrightarrow{\tau}
    (q_0,Z_0,Y_0')\xrightarrow{t_0}
    (q_1,Z_1,Y_1)\xrightarrow{\tau}
    (q_1,Z_1,Y_1')\xrightarrow{t_1}
    \cdots
  \end{equation*}
  
  where $Y_0=X$, $Y_i$ is determined by the transition, and $Y'_i=Y_i$
  for all $i \le j$ and for $i>j$, $Y_i' = \es$. Note that $\s'$ is a
  path in $\GZG(\Aa)$. For this to be true, each transition $(q_i,
  Z_i, Y_i') \xra{t_i} (q_{i+1}, Z_{i+1}, Y_{i+1})$ should be
  realizable from a valuation $\nu_i$ such that $\nu_i \sat (X - Y_i')
  > 0$. This is vacuously true if $i \le j$ since $Y'_i = X$ for all
  $i \le j$. For $i > j$, $Y'_i = \es$ and since $t_i$ does not
  contain a zero-check, the transition is realizable from a valuation
  $\nu_i$ in which all clocks are strictly greater than $0$.

  Since $\s$ is unblocked, $\s'$ is unblocked too. By definition all
  but finitely many nodes for $\s'$ are clear. Finally, $\s'$ visits
  an accepting node infinitely often. By Theorem
  \ref{thm:zeno_on_zero_guessing}, $\Aa$ has a non-Zeno run satisfying
  the B\"{u}chi condition.  \qed
\end{proof}

The above two observations give a sufficient condition for terminating
with a success when an SCC $\G$ with an accepting node is found in
$\ZG(\Aa)$. If the above two conditions do not hold, then $\G$ has no
clock that is reset and bounded from below (i.e. $x\geq 1$) and $\G$ either has
blocking clocks or zero-checks. If it has only blocking clocks, we apply the
procedure that restarts the exploration with blocking edges removed,
as described in Section~\ref{sec:algo:gzg}. If $\G$ has zero-checks,
we indeed use the guessing zone graph construction, however
\emph{restricted only to the nodes of $\G$}. The problem is to know
the initial set of clocks that need to be zero. We first define a few
notations.

Let $(q^\G, Z^\G)$ be the root of $\G$ as determined by Couvreur's
algorithm. Let $\GZG_{|\G}(\Aa)$ be the part of $\GZG(\Aa)$ rooted at
$(q^\G, Z^\G, X)$ and restricted only to the nodes and transitions
that occur in $\G$. We say that a run $\rho$ of $\Aa$ is
\emph{trapped} in an SCC $\G$ of $\ZG(\Aa)$ if a suffix of $\rho$ is
an instantiation of a path in $\G$.  The following lemma justifies the
use of the restricted guessing zone graph construction starting from
$(q^\G, Z^\G, X)$.

\begin{lemma}
  The automaton $\Aa$ has an accepting non-Zeno run trapped in an SCC
  $\G$ of $\ZG(\Aa)$ iff $\GZG_{|\G}$ has an SCC that is accepting,
  unblocked and contains a clear node.
\end{lemma}

\begin{proof} For the left-right direction, consider the following run
  $\rho$ of $\Aa$ trapped in $\G$:
  \begin{equation*}
    (q_0, \nu_0) \xra{\d_0, t_0} \dots (q_m, \nu_m)
    \xra{\d_m, t_m} \dots
  \end{equation*}
  where $q_m = q^\G$, $\nu_m \in Z^\G$ and $(q^\G, Z^\G)$ is the root
  of $\G$.  Consider the sequence $\s'$:
  \begin{equation*}
    (q_0,Z_0,Y_0)\xrightarrow{\tau}
    (q_0,Z_0,Y_0')\xrightarrow{t_0}
    (q_1,Z_1,Y_1)\xrightarrow{\tau}
    (q_1,Z_1,Y_1')\xrightarrow{t_1}
    \cdots
  \end{equation*}
  
  where
  \begin{itemize}
  \item $(q_0, Z_0)$ is the initial node of $\ZG(\Aa)$, the zone $Z_i$
    is determined by the transition $t_{i-1}$,
  \item $Y_0=X$, $Y_i$ is determined by the transition,
  \item $Y'_i=Y_i$ for all $i \le m$; for $i > m$, $Y_i' = \es$ if
    $\d_i > 0$ and $Y_i' = Y_i$ otherwise.
  \end{itemize}

  Observe that $Y_m = X$ and the suffix of $\s'$ starting from $(q_m,
  Z_m,Y_m)$ is a path of $\GZG_{|\G}(\Aa)$. Since there are infinitely
  many $i$ with $\d_i >0$, this suffix corresponds to an SCC that has
  a clear node. It is accepting and unblocked since the run $\rho$
  that we started with is accepting and non-Zeno.

  For the right-left direction, note that an accepting, unblocked SCC
  with a clear node in $\GZG_{|\G}(\Aa)$ corresponds to an accepting,
  unblocked path of $\GZG(\Aa)$ starting from $(q^\G, Z^\G, X)$ that
  visits a clear node infinitely often. It is straightforward to see
  that $(q^\G, Z^\G, X)$ is reachable from the initial node $(q_0,
  Z_0, X)$ of $\GZG(\Aa)$ through a path in which for all transitions
  $(q, Z, Y) \xra{\tau} (q', Z', Y')$, $Y' = Y$. Indeed, the
  restriction of $\GZG(\Aa)$ to its nodes with $Y=X$ is isomorphic to
  the zone graph $\ZG(\Aa)$. From this path of $\GZG(\Aa)$ and using
  Lemma \ref{lem:zero_guessing_to_TBA}, we can construct a accepting,
  non-Zeno run of $\Aa$ that is trapped in $\G$.  \qed
\end{proof}

Based on the above observations, we give the schema of the overall
optimized algorithm in Figure \ref{fig:fullalgo}. In the worst case,
the algorithm runs in time $\Oo(|\ZG(\Aa)|\cdot|X|^2)$. When the automaton
does not have zero-checks it runs in time
$\Oo(|\ZG(\Aa)|\cdot |X|)$. When the automaton further has no blocking
clocks, it runs in time $\Oo(|\ZG(\Aa)|)$.

\begin{figure}[t!]
  \centering
  \begin{tikzpicture}
    \draw (5,14) node {$\Aa$} (6,14); \draw[->] (5, 13.8) -- (5, 13);
    \draw (2, 12) rectangle (8, 13);
    \draw (5, 12.5) node {$\begin{array}{c}  \text{\tt{Compute}} ~ \ZG(\Aa) \\
        \texttt{using Couvreur's algorithm}
      \end{array}$} (7,12.5);
    \draw (1.5, 12.5) node {$\star$} (2,13);
    \draw[->] (8,12.5) -- node[above] {\tt{Finish}} (9.5,12.5) ;
    \draw (11, 12.5) node {$\Aa$ \tt{is empty}} (13,12.5);
    \draw[->] (5,12) -- (5, 10);
    \draw (7, 11) node {$\begin{array}{c} \texttt{Found SCC $\G$} \\  
        \texttt{with accepting node}\end{array}$} (9,11);
    \draw (2, 8) rectangle (8,10);
    \draw (5, 9.5) node {$\G$ \tt{is strongly non-Zeno?}} (8,9.5);
    \draw (5, 9) node {\tt{or is} $\G$} (6,9);
    \draw (5, 8.5) node {\tt{unblocked,free from
        zero-checks?}} (8,9.5);
    \draw[->] (8,9) -- node[above] {\tt{Yes}} (9.5,9);
    \draw (11, 9) node {$\Aa$ \tt{is non-empty}} (13,9);
    \draw[->] (5,8) -- node[right] {\tt{No}} (5, 7);
    \draw (5,6.5) node {\texttt{Is} $\G$ \texttt{maximal?}};
    \draw (3,7) rectangle (7,6);
    \draw[->] (7,6.5) -- node[above] {\tt{No}} (9.5,6.5);
    \draw (11, 6.5) node {\texttt{Continue} $\star$} (13,6.5);
    \draw[->] (5,6) -- node[right] {\tt{Yes}} (5, 5);
    \draw (5,4.5) node {$\G$ \texttt{has zero-checks?}};
    \draw (3,5) rectangle (7,4);
    \draw[->] (3,4.5) -- (2,4.5) -- node[left] {\texttt{No}} (2,3);
    \draw[->] (7,4.5) -- (8,4.5) -- node[left] {\texttt{Yes}} (8,3);
    \draw (0,1) rectangle (4,3); 
    \draw (2,2) node {$\begin{array}{c}\texttt{Is there a sub-SCC} \\
        \texttt{with accepting node \&} \\
        \texttt{no blocking clocks?} \end{array}$} (4,2);
    \draw (6,1) rectangle (10,3); 
    \draw (8,2) node {$\begin{array}{c}\GZG_{|\G}(\Aa)
        \texttt{ has SCC}  \\
        \texttt{with accepting node, } \\
        \texttt{clear node \& } \\
        \texttt{no blocking clocks?} \end{array}$} (10,2);
    \draw[->] (1, 1) -- node[left] {\tt{No}} (1,0) node[below]
    {\tt{Continue} $\star$};
    \draw[->] (9, 1) -- node[right] {\tt{No}} (9,0) node[below]
    {\tt{Continue} $\star$};
    \draw (3,1) -- node[right] {\tt{Yes}} (3,0);
    \draw (7,1) -- node[left] {\tt{Yes}} (7,0);
    \draw (3,0) -- (7,0);
    \draw[->] (5,0) -- (5,-1) node[below] {$\Aa$ \tt{is non-empty}} ;
  \end{tikzpicture}
  \caption{Algorithm to check for B\"{u}chi emptiness of
    $\Aa$. ``Continue'' loops back to computing $\ZG(\Aa)$ using
    Couvreur's Algorithm.}
  \label{fig:fullalgo}
\end{figure}


%
\section{Experiments}
\label{sec:experiments-conclusion}

We have implemented our algorithms in a prototype verification
tool. Given a \emph{network} $\Aa_1, \dots, \Aa_n$ of timed B\"{u}chi
automata, we want to check if this network satisfies a property $\phi$
specified in some logic. We consider a property $\phi$ that can be
translated into a timed automaton $\Aa_{\neg\phi}$ such that the
network satisfies $\phi$ iff the product timed automaton $\Aa_1 \times
\dots \Aa_n \times \Aa_{\neg\phi}$ has an empty
language. Table~\ref{tbl:experiments} presents the results that we
obtained on several classical examples. The ``Models'' column
represents the product of the network Timed Büchi Automata and the
property to verify. We give the number of processes in the network for
each model. A tick in the ``Sat.'' columns tells that the property is
satisfied by the model. The ``Zone Graph'' column gives the number of
nodes in the zone graph. Next, for the ``Strongly non-Zeno''
construction, we give the size of the resulting zone graph followed by
the number of nodes that are visited during verification using the
Couvreur's algorithm. Similarly for the ``Guessing Zone Graph'' but
using the algorithm in section~\ref{sec:algo:gzg}. Finally, the last
column corresponds to our fully optimized algorithm as described in
section~\ref{sec:algo:optimized}.

We have considered three types of properties: reachability properties
(mutual exclusion, collision detection for CSMA/CD), liveness
properties (access to the resource infinitely often), and bounded
response properties (which are reachability properties with real-time
requirements). Reachability properties require to find a path
to a target state starting from the initial state. Although this
path is a finite sequence, it is realistic only if this finite
sequence can be extended to a 
non-Zeno path of the automaton. Therefore, while verifying
reachability properties, we check if the automaton has a non-Zeno path
that contains the target state. 

The strongly non-Zeno construction outperforms the guessing zone graph
construction for reachability properties. This is particularly the
case for mutual exclusion on the Fischer's protocol and collision
detection for the CSMA/CD protocol. For liveness properties, the
results are more balanced. On the one hand, the strongly non-Zeno
construction is once again more efficient for the CSMA/CD protocol. On
the other hand the differences are tight in the case of Fischer
protocol. The guessing zone graph construction distinguishes itself
for bounded response properties. Indeed, the Train-Gate model is an
example of exponential blowup for the strongly non-Zeno construction.

We notice that on-the-fly algorithms perform well. Even when the
graphs are big, particularly in case when automata are not empty, the
algorithms are able to conclude after having explored only a small
part of the graph. Our optimized algorithm outperforms the two others
on most examples. Particularly, for the CSMA/CD protocol with 5
stations our algorithm needs to visit only 4841 nodes while the two
other methods visited 8437 and 21038 nodes. This confirms our initial
hypothesis: most of the time, the zone graph contains enough
information to ensure time progress. As a consequence, checking
non-Zenoness and emptiness is done at the same cost as checking
emptiness only. This is in turn achieved at a cost that is similar to
reachability checking.

Our optimization using lower bounds on clocks also proves useful for
the FDDI protocol example. One of its processes has zero-checks, but
since some other clock is bounded from below and reset, it was not
necessary to explore the guessing zone graph to conclude
non-emptiness.

\begin{table}
{
  \scriptsize
  \begin{center}
    \begin{tabular}{|l|c|r||r|r||r|r||r|}
      \hline
      \multirow{2}{*}{Models ($\Aa$)} &
      \multirow{2}{*}{Sat.} &
      $\ZG(\Aa)$ &
      \multicolumn{2}{|c||}{$\ZG(SNZ(\Aa))$} &
      \multicolumn{2}{|c||}{$\GZG(\Aa)$} &
      Optimized\\
      \cline{3-8}
      &
      &
      \multicolumn{1}{|c||}{size} &
      \multicolumn{1}{|c|}{size} & \multicolumn{1}{|c||}{visited} &
      \multicolumn{1}{|c|}{size} & \multicolumn{1}{|c||}{visited} &
      \multicolumn{1}{|c|}{visited}\\
      \hline
      Train-Gate2 (mutex) & $\surd$ & 134 & 194 & 194 & 400 & 400 &
      134\\
      Train-Gate2 (bound. resp.) & & 988 & 227482 & 352 & 3840 & 1137
      & 292\\
      Train-Gate2 (liveness) & & 100 & 217 & 35 & 298 & 53 & 33\\
      \hline
      Fischer3 (mutex) & $\surd$ & 1837 & 3859 & 3859 & 7292 & 7292 &
      1837\\
      Fischer4 (mutex) & $\surd$ & 46129 & 96913 & 96913 & 229058 &
      229058 & 46129\\
      Fischer3 (liveness) & & 1315 & 4962 & 52 & 5222 & 64 & 40\\
      Fischer4 (liveness) & & 33577 & 147167 & 223 & 166778 & 331 &
      207\\
      \hline
      FDDI3 (liveness) & & 508 & 1305 & 44 & 3654 & 79 & 42\\
      FDDI5 (liveness) & & 6006 & 15030 & 90 & 67819 & 169 & 88\\
      FDDI3 (bound. resp.) & & 6252 & 41746 & 59 & 52242 & 114 & 60\\
      \hline
      CSMA/CD4 (collision) & $\surd$ & 4253 & 7588 & 7588 & 20146 &
      20146 & 4253\\
      CSMA/CD5 (collision) & $\surd$ & 45527 & 80776 & 80776 & 260026
      & 260026 & 45527\\
      CSMA/CD4 (liveness) & & 3038 & 9576 & 1480 & 14388 & 3075 &
      832\\
      CSMA/CD5 (liveness) & & 32751 & 120166 & 8437 & 186744 & 21038 &
      4841\\
      \hline
    \end{tabular}
  \end{center}
}
\caption{Experimental Results. The ``Sat.'' column tells which
  properties are satisfied by the model. The ``size'' columns give the
  number of nodes in the corresponding graphs. The ``visited'' columns
  give the number of nodes that are visited by the corresponding
  algorithm. The results correspond to the Couvreur's algorithm for
  $\ZG(SNZ(\Aa))$, the algorithm in Section~\ref{sec:algo:gzg} for
  $\GZG(\Aa)$ and the algorithm in Section~\ref{sec:algo:optimized}
  for the ``Optimized'' column.}
\label{tbl:experiments}
\end{table}

\section{Conclusions}

The Büchi non-emptiness problem is one of the standard problems for
timed automata. Since the paper introducing the model, it has been
widely accepted that the addition of one auxiliary clock is an
adequate method to deal with the problem of Zeno paths. This technique
is also used in the recently proposed zone based algorithm for the
problem~\cite{T:TOCL:2009}.

In this paper, we have shown that in some cases the auxiliary clock
may cause exponential blowup in the size of the zone graph. We have
proposed another method that is based on a modification of the zone
graph. The resulting graph grows only by a factor that is linear in
the number of clocks. In our opinion, the efficiency gains of our
method outweigh the fact that it requires some small modifications in
the code dealing with zone graph exploration. Moreover, liveness can
be checked at the same cost as reachability as demonstrated by our
experiments. This also shows that in most cases the zone graph already
contains enough information to handle non-Zenoness.

As future work we plan to extend our algorithm to commonly used
syntactic extensions of timed automata. For example, UPPAAL and Kronos
allow reset of clocks to arbitrary values, which is convenient for
modeling real life systems. This would require to extend the guessing
zone graph construction and consequently our algorithm. In this paper,
we considered the $\Approx$ abstraction that has been largely improved
by later works~\cite{Behrmann:STTT:2006}. It has been shown that these
new abstractions preserve Büchi conditions~\cite{Li:FORMATS:2009}. We
plan to study the extension of our technique to these
abstractions. Finally, we also plan to extend our construction to
extract non-Zeno strategies in timed games.


\bibliographystyle{plain}
\bibliography{m}

\end{document}